\documentclass[letterpaper,journal,twocolumn]{IEEEtran}
\usepackage{latexsym}
\usepackage{amsmath}
\usepackage{amssymb}
\usepackage{amsthm}
\usepackage{cite}
\usepackage{graphicx}

\newtheorem{theorem}{Theorem}
\newtheorem{lemma}{Lemma}
\newtheorem{proposition}{Proposition}
\newtheorem{corollary}{Corollary}
\newtheorem{remark}{Remark}
\newtheorem{example}{Example}
\newtheorem{definition}{Definition}

\newcommand{\Alg}{\ensuremath{\mathcal A}}
\newcommand{\Fun}{\ensuremath{\mathcal F}}
\newcommand{\Ora}{\ensuremath{\mathcal O}}
\newcommand{\Pro}{\ensuremath{\mathcal P}}
\newcommand{\cB}{\ensuremath{\mathcal B}}
\newcommand{\cU}{\ensuremath{\mathsf U}}
\newcommand{\cX}{\ensuremath{\mathsf X}}
\newcommand{\cY}{\ensuremath{\mathsf Y}}
\newcommand{\cZ}{\ensuremath{\mathsf Z}}

\newcommand{\Algs}{\ensuremath{\mathfrak A}}

\newcommand{\Prob}{\ensuremath{{\mathbb P}}}

\newcommand{\ProbQ}{\ensuremath{{\mathbb Q}}}
\newcommand{\E}{\ensuremath{\mathbb E}}

\newcommand{\var}{\operatorname{var}}
\newcommand{\ind}[1]{{\bf 1}_{\{#1\}}}

\newcommand{\Naturals}{\ensuremath{\mathbb N}}
\newcommand{\Reals}{\ensuremath{\mathbb R}}

\def\disc{d}
\def\Dim{n}

\def\eps{\varepsilon}

\def\err{\operatorname{err}}
\def\Err{\overline{\err}}

\def\Lip{{\rm Lip}}
\def\SC{{\rm SC}}

\newcommand{\tr}{\ensuremath{{\scriptscriptstyle\mathsf{T}}}}
\def\Normal{{\mathcal N}}

\newcommand{\deq}{\triangleq}
\def\wh#1{\ensuremath{\hat{#1}}}

\def\vol{\operatorname{vol}}
\def\argmin{\operatornamewithlimits{arg\,min}}

\begin{document}

\title{Information-Based Complexity, Feedback\\
and Dynamics in Convex Programming}

\author{Maxim Raginsky,~\IEEEmembership{Member,~IEEE},
and Alexander Rakhlin
\thanks{The work of M. Raginsky was supported in part by the NSF under grant CCF-1017564. The work of A. Rakhlin was supported by the NSF CAREER award DMS-0954737. A preliminary version of this work was presented at the 47th Annual Allerton Conference on Communication, Control and Computing, Monticello, IL, September/October 2009.}

\thanks{M.~Raginsky is with the Department of Electrical and Computer Engineering, Duke University, Durham, NC 27708, USA. E-mail: m.raginsky@duke.edu.}
\thanks{A.~Rakhlin is with the Department of Statisics, Wharton School of Business, University of Pennsylvania, Philadelphia, PA 19104, USA. E-mail: rakhlin@wharton.upenn.edu.}
}

\maketitle
\thispagestyle{empty}

\begin{abstract}
	We study the intrinsic limitations of sequential convex optimization through the lens of feedback information theory. In the oracle model of optimization, an algorithm queries an {\em oracle} for noisy information about the unknown objective function, and the goal is to (approximately) minimize every function in a given class using as few queries as possible. We show that, in order for a function to be optimized, the algorithm must be able to accumulate enough information about the objective. This, in turn, puts limits on the speed of optimization under specific assumptions on the oracle and the type of feedback. Our techniques are akin to the ones used in statistical literature to obtain minimax lower bounds on the risks of estimation procedures; the notable difference is that, unlike in the case of i.i.d.\ data, a sequential optimization algorithm can gather observations in a {\em controlled} manner, so that the amount of information at each step is allowed to change in time. In particular, we show that optimization algorithms often obey the law of diminishing returns: the signal-to-noise ratio drops as the optimization algorithm approaches the optimum. To underscore the generality of the tools, we use our approach to derive fundamental lower bounds for a certain active learning problem. Overall, the present work connects the intuitive notions of ``information'' in optimization, experimental design, estimation, and active learning to the quantitative notion of Shannon information.
\\ \\
\end{abstract}

\section{Introduction}

\PARstart{M}{any} problems arising in such areas as communications and signal processing, contrtol, machine learning, economics, and many others require solving mathematical programs of the form
\begin{align}\label{eq:conv_program}
\min \{ f(x) : x \in \cX \},
\end{align}
where $f : \Reals^\Dim \to \Reals$ is a convex objective function and $\cX$ is a compact, convex subset of $\Reals^\Dim$. Therefore, it is important to have a clear understanding of the {\em fundamental limits} on the efficiency of convex programming methods.

A systematic study of these fundamental limits was initiated in the 1970's by Nemirovski and Yudin \cite{NemYud83}. In their framework, an optimization algorithm is a sequential procedure that repeatedly queries a black-box {\em oracle} for information about the function being optimized, each query depending on the past information. The oracle may be deterministic (for example, giving the value of the function and its derivatives up to some order at any point) or stochastic. This leads to the notion of {\em information-based complexity}, i.e., the smallest number of oracle calls needed to minimize any function in a given class to a desired accuracy. The results in \cite{NemYud83} are very wide in scope and cover a variety of convex programming problems in Banach spaces; finite-dimensional versions are covered in \cite{Nes04} and \cite{NJLS09}.

For deterministic oracles, Nemirovski and Yudin derived lower bounds on the information complexity of convex programming using a ``counterfactual" argument: given any algorithm that purports to optimize all functions in some class $\Fun$ to some degree of accuracy $\eps$ using at most $T$ oracle calls, one explicitly constructs, for a particular history of queries and oracle responses, a function in $\Fun$ which is consistent with this history, and yet cannot be $\eps$-minimized by the algorithm using fewer than $T$ oracle calls (see also \cite{Nes04}). A similar approach was also used for stochastic oracles.

Proper application of this {\em method of resisting oracles} requires a lot of ingenuity. In particular, the stochastic case involves fairly contrived noise models, unlikely to be encountered in practice. In this paper, which expands upon our preliminary work \cite{RagRak09}, we will show that the same (and many other) lower bounds can be derived using a much simpler information-theoretic technique reminiscent of the way one proves minimax lower bounds in statistics \cite{Yu97,YanBar99,Tsy09}. Namely, we reduce optimization to hypothesis testing with controlled observations and then relate the resulting probability of error to information complexity using Fano's inequality and a series of mutual information bounds. These bounds highlight the role of {\em feedback} in choosing the next query based on past observations. One notable feature of our approach is that it does not require constructing particularly ``strange" functions or noise models. Moreover, we derive a ``law of diminishing returns" for a wide class of convex optimization schemes, which says that the decay of optimization error is offset by the decay of the rate at which the algorithm can reduce its uncertainty about the objective function.

The idea of relating optimization to hypothesis testing is not new. For instance, Shapiro and Nemirovski \cite{ShaNem05} derive a lower bound on the information complexity of a certain class of one-dimensional linear optimization problems by reducing optimization to a binary hypothesis testing problem pertaining to the parameter of a Bernoulli random variable (the outcome of a coin toss). The reduction consists in showing that any good optimization algorithm can be converted into an accurate estimator of the coin bias based on repeated independent trials; then one can derive the lower bound on the information complexity (equivalently, the minimum necessary number of coin tosses) from the data processing inequality for divergence (or Fano's inequality). This approach was recently extended to multidimensional optimization problems by Agarwal \emph{et al.}~\cite{Aga09,Aga10}. Like the present paper, their work uses information-theoretic methods to derive lower bounds on the oracle complexity of convex optimization, and their results are qualitatively similar to some of ours. However, what sets our work apart from \cite{ShaNem05,Aga09,Aga10} is that we explicitly account for the controlled manner in which the algorithm interacts with the oracle. This, in turn, allows us to derive tight lower bounds on the rate of error decay for certain types of \emph{infinite-step} descent algorithms, which is not possible with the reduction to coin tossing.

Sequential procedures have become increasingly popular in the field of machine learning, mostly due to the abundance of data and the resulting need to perform computation on-line. Convex optimization is not the only sequential setting being studied: recent research in machine learning has also focused on such scenarios as active learning, multi-armed bandits, and experimental design, to name a few. In all these settings, one element is common: each additional ``action'' should provide additional ``information'' about some unknown quantity. Translating this intuitive notion of ``information'' into precise information-theoretic statements is often difficult. Our contribution consists in offering such a translation for convex optimization and closely related problems.

\subsection{Notation}

Given a continuous function $f : \cX \to \Reals$ on a compact domain $\cX \subset \Reals^\Dim$, we denote by $f^*$ its minimum value over $\cX$:
\begin{align*}
f^* = \inf_{x \in \cX} f(x).
\end{align*}
We will use several basic notions from nonsmooth convex analysis \cite{Hiriart}. The {\em subdifferential} of $f$ at $x$, denoted by $\partial f(x)$, is the set of all $g \in \Reals^\Dim$, such that
\begin{align*}
f(y) \ge f(x) + g^\tr (y-x), \qquad \forall y \in \Reals^\Dim.
\end{align*}
Any such $g$ is a {\em subgradient} of $f$ at $x$. For a convex $f$, the subdifferential $\partial f(x)$ is always nonempty. When $|\partial f(x)|=1$, its only element is precisely the gradient $\nabla f(x)$. By $\| x \|_p$ we denote the $\ell_p$ norm of $x \in \Reals^\Dim$; the $\ell_2$ norm will also be denoted by $\| \cdot \|$. By $B^\Dim_p$ we denote the unit ball in $\Reals^\Dim$ in the $\ell_p$ norm. The $\ell_2$-diameter of $\cX$  is defined as
\begin{align*}
D_\cX \deq \sup_{x,x' \in \cX} \| x - x' \|.
\end{align*}
The $n \times n$ identity matrix will be denoted by $I_n$.

All abstract spaces are assumed to be standard Borel (i.e.,~Borel subsets of a complete separable metric space), and will be equipped with their Borel $\sigma$-fields. If $\cZ$ is such a space, then $\cB_\cZ$ will denote the corresponding $\sigma$-field. All functions between such spaces are assumed to be measurable. If $\cZ_1$ and $\cZ_2$ are two such spaces, then a {\em Markov kernel} \cite{Shiryaev,Kallenberg} from $\cZ_1$ to $\cZ_2$ is a mapping $P : \cB_{\cZ_2} \times \cZ_1 \to [0,1]$, such that for any $z_1 \in \cZ_2$ $P(\cdot|z_1)$ is a probability measure on $(\cZ_2,\cB_{\cZ_2})$ and for any $B \in \cB_{\cZ_2}$ $P(B|\cdot)$ is a measurable function on $\cZ_1$. We will use the standard notation $P(dz_2|z_1)$ for such a kernel.

We will work with the usual information-theoretic quantities, which are well-defined in standard Borel spaces \cite{GrayIT}. Given two (Borel) probability measures $\Prob$ and $\ProbQ$ on $\cZ$, their {\em divergence} is
\begin{align*}
	D(\Prob \| \ProbQ) \deq \begin{cases}
	 \displaystyle \int_\cZ \left( \log \frac{d\Prob}{d\ProbQ}\right) d\Prob, & \text{if } \Prob \ll \ProbQ\\
	+ \infty, & \text{otherwise}
\end{cases}
\end{align*}
where the notation $\Prob \ll \ProbQ$ means that $\Prob$ is {\em absolutely continuous} w.r.t.\ $\ProbQ$, i.e., $\ProbQ(B) = 0$ for any $B \in \cB_\cZ$ implies that $\Prob(B) = 0$ as well. If $\cZ$ is a product space, $\cZ = \cZ_1 \times \cZ_2$, then the {\em conditional divergence} between two probability distributions $\Prob$ and $\ProbQ$ on $\cZ$ given $\Prob_{Z_1}$ (the $\cZ_1$-marginal of $\Prob$) is
\begin{align}\label{eq:cond_div}
&	D(\Prob_{Z_2|Z_1} \| \ProbQ_{Z_2|Z_1} | \Prob_{Z_1}) \nonumber\\
&\deq \int_{\cZ_1}\Prob_{Z_1}(dz_1) D\big(\Prob_{Z_2|Z_1}(\cdot|z_1) \big\| \ProbQ_{Z_2|Z_1}(\cdot|z_1)\big),
\end{align}
where $\Prob_{Z_2|Z_1}$ and $\ProbQ_{Z_2|Z_1}$ are any versions of the regular conditional probability distributions of $Z_2$ given $Z_1$ under $\Prob$ and $\ProbQ$, respectively. This definition extends in the obvious way to situations when $\cZ_1$ or $\cZ_2$ are themselves product spaces. Thus, if $\Prob$ and $\ProbQ$ are two probability distributions for a random triple $(Z_1,Z_2,Z_3)$ taking values in a product space $\cZ = \cZ_1 \times \cZ_2 \times \cZ_3$, such that, under $\ProbQ$, $Z_2$ and $Z_3$ are conditionally independent given $Z_1$, i.e., $\ProbQ_{Z_3|Z_1,Z_2} = \ProbQ_{Z_3|Z_1}$ $\ProbQ$-a.s., then we will write
\begin{align}\label{eq:cond_div_CI}
& D(\Prob_{Z_3|Z_1,Z_2} \| \ProbQ_{Z_3|Z_1,Z_2} | \Prob_{Z_1,Z_2}) \nonumber\\
& \qquad = D(\Prob_{Z_3|Z_1,Z_2} \| \ProbQ_{Z_3|Z_1} | \Prob_{Z_1,Z_2}).
\end{align}	
Given a random couple $(Z_1,Z_2) \in \cZ$ with probability distribution $\Prob$, the {\em mutual information} between $Z_1$ and $Z_2$ is
\begin{align*}
	I(Z_1; Z_2) \deq D( \Prob \| \Prob_{Z_1} \otimes \Prob_{Z_2}) \equiv D( \Prob_{Z_2|Z_1} \| \Prob_{Z_2} | \Prob_{Z_1}).
\end{align*}
Given a random triple $(Z_1,Z_2,Z_3) \in \cZ_1 \times \cZ_2 \times \cZ_3$, the {\em conditional mutual information} between $Z_2$ and $Z_3$ given $Z_1$ is
\begin{align}
	I(Z_2; Z_3 | Z_1) &\deq D\big( \Prob_{Z_2,Z_3|Z_1} \big \| \Prob_{Z_2|Z_1} \otimes \Prob_{Z_3|Z_1} \big | \Prob_{Z_1} \big) \nonumber \\
	&\equiv D \big( \Prob_{Z_3|Z_1,Z_2} \big\| \Prob_{Z_3|Z_1} \big| \Prob_{Z_1,Z_2} \big), \label{eq:cond_MI}
\end{align}
where \eqref{eq:cond_MI} follows from Bayes' rule and from \eqref{eq:cond_div_CI}. In other words, the conditional mutual information $I(Z_2; Z_3 |Z_1)$ is given by the conditional divergence between the joint distribution of $Z_1,Z_2,Z_3$ and the distribution under which $Z_2$ and $Z_3$ are conditionally independent given $Z_1$.

\section{Sequential optimization algorithms and their information-based complexity}
\label{sec:IBC}

The work of Nemirovski and Yudin \cite{NemYud83} deals with fundamental limitations of sequential optimization algorithms in the real-number model of computation. The basic setting is as follows. We have a class $\Fun$ of convex functions $f: \cX \to \Reals$ on some compact convex domain $\cX \subset \Reals^\Dim$. We seek an ``optimal" algorithm that would solve the optimization problem \eqref{eq:conv_program} with a given guarantee of accuracy regardless of which $f \in \Fun$ were to be optimized. The algorithms of interest operate by repeatedly querying an {\em oracle} for information about the unknown objective $f$ at appropriately selected points in $\cX$ and then combining the accumulated information to form a solution. The notion of optimality of an algorithm pertains to the number of queries it makes before producing a solution, without regard to the {\em combinatorial complexity} of computing each query. In other words, we are interested in the {\em information-based complexity} (IBC) \cite{Traub,Plaskota} of convex optimization problems.

The theory of IBC is concerned with intrinsic difficulty of computational problems in terms of the minimum amount of information needed to solve every problem in a given class with a given guarantee of accuracy. The word ``information" here does not refer to information in the sense of Shannon, but rather to what is known {\em a priori} about the problem being solved, as well as what an algorithm is allowed to {\em learn} during its operation. There are three aspects inherent in this notion of information --- it is {\em partial}, {\em noisy}, and {\em priced}. Let us explain informally what these three terms mean in the context of optimization by means of a simple example.

Let $\cX = [0,1]$, and consider the function class
\begin{align}\label{eq:simple_function_class}
\Fun = \left\{ f_\theta(x) \deq \frac{1}{2} |x-\theta|^2 : \theta \in \cX \right\}.
\end{align}
We wish to design an algorithm that minimizes every $f = f_\theta \in \Fun$ to a given accuracy $\eps > 0$. At the outset, the only {\em a priori} information available to the algorithm consists of the problem domain $\cX$, the function class $\Fun$, and the desired accuracy $\eps$. The algorithm is allowed to query the value and the derivative of $f$ at any finite set of points $\{x_1,\ldots,x_T\} \subset \cX$ before arriving at a solution, which we denote by $x_{T+1}$. The queries are answered by an {\em oracle}, i.e., a (possibly stochastic) device that knows the function $f$ (or, equivalently, the parameter $\theta$) and responds to any {\em query} $x \in \cX$ with $Y(\theta,x,\omega) \in \Reals^2$, where $\omega$ is a random element from some probability space $(\Omega,\cB,\Prob)$ that represents oracle noise. The random variable $Y(\theta,x,\omega)$ is assumed to be a noisy observation of the pair $(f_\theta(x),f'_\theta(x))$. For concreteness, let us suppose that
\begin{align}
Y(\theta,x,\omega) &= Y(\theta,x,(W,Z)) \nonumber \\
&= ( f_\theta(x) + W, f'_\theta(x) + Z) \nonumber \\
&= \left( \frac{1}{2} | x - \theta|^2 + W, x - \theta + Z \right), \label{eq:simple_Gaussian_oracle}
\end{align}
where $W$ and $Z$ are an i.i.d.\ pair of $\Normal(0,\sigma^2)$ random variables.

The interaction of the algorithm and the oracle takes place as follows. Let $\{(W_t,Z_t)\}^\infty_{t=1}$ be an i.i.d.\ sequence. At time $t = 1,2,\ldots$, the algorithm computes the query $X_t$ as a function of the past queries $X_\tau, 1\le\tau\le t-1$ and the corresponding oracle responses $Y_\tau = Y(\theta,X_\tau,(W_\tau,Z_\tau)), 1 \le \tau \le t-1$. At time $t=0$ the algorithm knows only that $f \in \Fun$; this represents the {\em a priori} information. At time $t \ge 1$, the algorithm acquires additional data $(X^t,Y^t) = ((X_1,Y_1),\ldots,(X_t,Y_t))$, and so can refine its {\em a priori} information. At every time step, the information is {\em partial} in the sense that there are (potentially infinitely) many functions consistent with it, and it is also {\em noisy} due to the presence of the additive disturbances $(W^t,Z^t)$.

Formally, for the example outlined above, an algorithm that makes $T$ queries (or a $T$-step algorithm) is a tuple $\Alg = \{ \Alg_t : \cX^{t-1} \times \cY^{t-1} \to \cX \}^{T+1}_{t=1}$, where $\cY = \Reals^2$, so that, for $1 \le t \le T$, $X_t = \Alg_t(X^{t-1},Y^{t-1})$ is the query at time $t$, and $X_{T+1} = \Alg_{T+1}(X^T,Y^T)$ is the solution. We assume that information is {\em priced} in the sense that the algorithm is charged some fixed cost $c > 0$ for every query it makes. Thus, it is desired to keep the number of queries to a minimum. With this in mind, we can define the IBC for a given accuracy $\eps > 0$ as
\begin{align*}
&\text{IBC}(\eps) = \inf \Bigg\{ T \ge 1: \exists \Alg = \{\Alg_t\}^{T+1}_{t=1} \nonumber\\
& \qquad \qquad \text{ s.t. } \sup_{\theta \in \cX} \left[\E f_\theta(X_{T+1}) - f^*_\theta\right] \le \eps \Bigg\},
\end{align*}
where the expectation is taken w.r.t.\ the noise process $\{(W_t,Z_t)\}^\infty_{t=1}$. For this particular problem it can be shown that
\begin{align}\label{eq:IBC_bound}
\text{IBC}(\eps) &= \begin{cases}
1, & \sigma^2 = 0, \eps \in [0,1/2) \\
\Theta \left( \frac{\sigma^2}{\eps}\right), & \sigma^2 > 0, \eps \in (0,1/2) \\
0, &  \eps \ge 1/2
\end{cases}
\end{align}
The first entry $(\sigma^2=0,0 \le \eps < 1/2)$ follows because the algorithm can just query $x_1 = 0$, obtain the response $y_1 = ((1/2)\theta^2, -\theta)$, and immediately compute $x_2 = \theta$; the last entry $(\sigma^2 > 0 ,\eps \ge 1/2)$ follows because the maximum value of any $f_\theta \in \Fun$ on $\cX$ is at most $1/2$. The intermediate regime $(\sigma^2 > 0, \eps \in (0,1/2))$ is more involved. The main contribution of the present paper is a unified information-theoretic framework for deriving lower bounds on the IBC of arbitrary sequential algorithms for solving convex programming problems.

\subsection{Formal definitions}
\label{ssec:formal_defs}

The above discussion can be formalized as follows:

\begin{definition} A {\em problem class} is a triple $\Pro = (\cX, \Fun, \Ora)$ consisting of the following objects:
\begin{enumerate}
\item A compact, convex {\em problem domain} $\cX \subset \Reals^\Dim$; \item An {\em instance space} $\Fun$, which is a class of convex functions $f : \cX \to \Reals$;
\item An {\em oracle} $\Ora= (\cY, P)$, where $\cY$ is the {\em oracle information space} and  $P(dy|f,x), dy \in \cB_\cY, f \in \Fun, x \in \cX$, is a Markov kernel\footnote{Recall that $\Fun$ is a subset of $C(\cX)$, the space of all continuous real-valued functions on $\cX$. Equipped with the usual sup norm, $C(\cX)$ is a separable Banach space, so a Markov kernel from $\Fun \times \cX$ into $\cY$ is well-defined.}.
\end{enumerate}
\end{definition}

\noindent Some restrictions must be imposed in order to exclude oracles that are ``too informative," an extreme example being $\cY = \Fun \times \cX$ and $P(dy|f,x) = \delta_{f,x}(dy)$. One way to rule this out is to require the oracle in question to be {\em local} \cite{NemYud83}:

\begin{definition}\label{def:local_oracle}
We say that an oracle $\Ora$ is {\em local} if for every $x \in \cX$ and every pair $f,f' \in \Fun$ such that $f = f'$ in some open neighborhood of $x$, we have
\begin{align*}
P(dy|f,x) = P(dy|f',x), \qquad \forall dy \in \cB_\cY.
\end{align*}
\end{definition}

\noindent It is easy to see that the oracle described right before the definition is not local. Indeed, fix a point $x \in \cX$ and consider any two functions $f,f' \in \Fun$ that agree on some open neighborhood of $x$, but are not equal outside this neighborhood. Then $P(dy|f,x) = \delta_{f,x}(dy)$, but $P(dy|f',x) = \delta_{f',x}(dy)$, which violates locality. Most oracles encountered in practice are local (see, for instance, the examples in Section~\ref{sec:examples}).

To gain more insight into stochastic oracles, we can appeal to the basic structural result for Markov kernels: If $\cZ_1$ and $\cZ_2$ are standard Borel spaces, then any Markov kernel $P(dz_2|z_1)$ from $\cZ_1$ to $\cZ_2$ can be realized in the form $Z_2 = \Phi(z_1,W)$, where $W$ is a random variable uniformly distributed on $[0,1]$ and $\Phi : \cZ_1 \times [0,1] \to \cZ_2$ is a measurable mapping \cite[Lemma~3.22]{Kallenberg}. Thus, for any stochastic oracle $P(dy|f,x)$ we can find a {\em deterministic} oracle $\psi : \Fun \times \cX \to \cU$ with some information space $\cU$ and a measurable mapping $\Phi : \cU \times [0,1] \to \cY$, such that $P$ can be realized as \begin{align}\label{eq:SLO}
Y = \Phi(\psi(f,x),W)
\end{align}
with $W$ as above. Thus, $P$ will be local in the sense of Definition~\ref{def:local_oracle} whenever its ``deterministic part" $\psi$ is local.

Next, we make the notion of an optimization algorithm precise. In this paper, we deal only with deterministic algorithms, although all the results can be easily extended to cover randomized algorithms as well (cf.~\cite{NemYud83} for details):

\begin{definition}A {\em $T$-step algorithm} for a given $\Pro = (\cX,\Fun,\Ora)$ is a sequence of mappings $\Alg = \{\Alg_t : \cX^{t-1} \times \cY^{t-1} \to \cX \}^{T+1}_{t=1}$. The set of all $T$-step algorithms for $\Pro$ will be denoted by $\Algs_T(\Pro)$.
\end{definition}

\noindent The interaction of any $\Alg \in \Algs_T(\Pro)$ with $\Ora$, shown in Figure~\ref{fig:feedback_optimization_new}, is described recursively as follows:

\begin{enumerate}

\item At time $t=0$, a problem instance $f \in \Fun$ is selected by Nature and revealed to $\Ora$, but not to $\Alg$.

\item At each time $t = 1,2,\ldots,T$:
\begin{itemize}
\item  $\Alg$ queries $\Ora$ with $X_t = \Alg_t(X^{t-1},Y^{t-1})$, where $(X_\tau,Y_\tau) \in \cX \times \cY$ is the algorithm's query and the oracle's response at time $\tau \le t-1$.
\item $\Ora$ responds with a random element $Y_t \in \cY$ according to $P( dY_t | f,X_t)$.
\end{itemize}
\item At time $t = T+1$, $\Alg$ outputs the candidate minimizer $X_{T+1} = \Alg_{T+1}(X^T,Y^T)$.
\end{enumerate}

\begin{figure}[htbp]
	\centering
		\includegraphics[width=0.8\columnwidth]{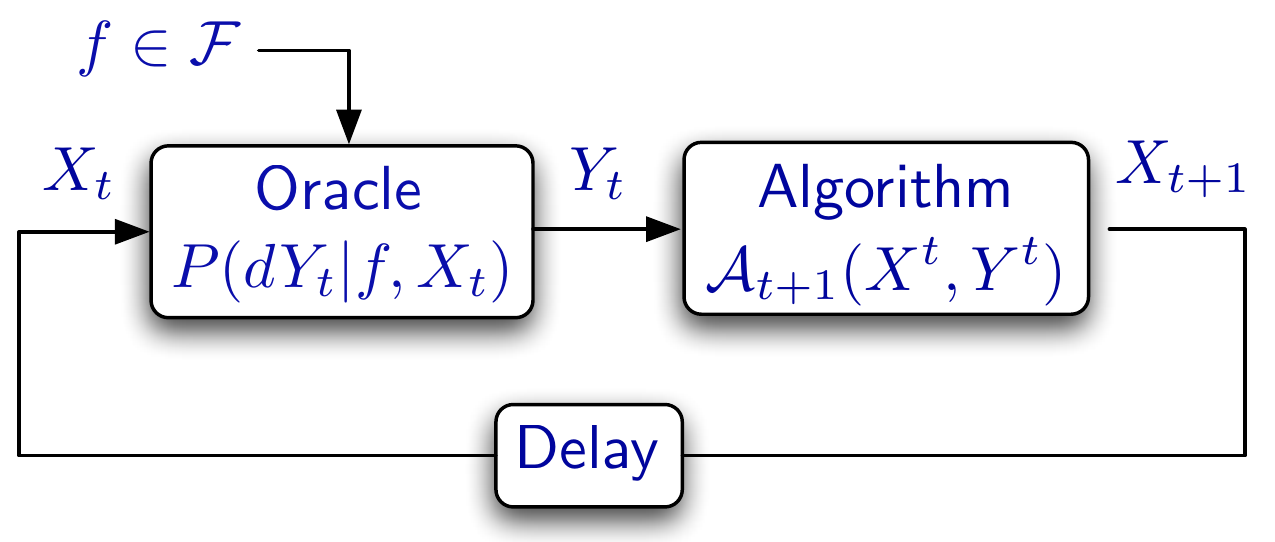}
	\caption{\label{fig:feedback_optimization_new} Interaction of an algorithm $\Alg$ and an oracle $\Ora$.}
\end{figure}

\noindent We can view the set-up of Figure~\ref{fig:feedback_optimization_new} as a discrete-time {\em stochastic dynamical system} with an unknown ``parameter'' $f \in \Fun$, input sequence $\{X_t\}$, and output sequence $\{Y_t\}$. The objective is to drive the system as quickly as possible to an $\eps$-minimizing state, i.e., any $x \in \cX$ such that $f(x) - f^* < \eps$, for every $f \in \Fun$. We are interested in the {\em fundamental limits} on the speed with which this can be done. Defining the {\em error} of $\Alg \in \Algs_T(\Pro)$ on $f \in \Fun$ by
\begin{align*}
\err_\Alg(T,f) \deq f(X_{T+1}) - \inf_{x \in \cX} f(x) = f(X_{T+1}) - f^*,
\end{align*}
we introduce the following definition:

\begin{definition} Fix a problem class $\Pro = (\cX,\Fun,\Ora)$. For any $r \ge 1$, $\eps > 0$, and $\delta \in (0,1)$, we define the $r$th-order {\em $(\eps,\delta)$-complexity} and the {\em $\eps$-complexity} of $\Pro$, respectively, as
\begin{align*}
& K^{(r)}_{\Pro}(\eps,\delta) \deq \inf \Big\{ T \ge 1: \exists \Alg \in \Algs_T(\Pro) \nonumber\\
& \qquad \qquad \emph{\,\,s.t.\,} \sup_{f \in \Fun} \Pr\big(\err^r_\Alg(T,f) \ge \eps\big) \le \delta \Big\}; \\
& K^{(r)}_{\Pro}(\eps) \deq  \inf \Big\{ T \ge 1: \exists \Alg \in \Algs_T(\Pro)  \nonumber\\
& \qquad \qquad \emph{\,\,s.t.\,} \sup_{f \in \Fun} \E \err^r_\Alg(T,f) < \eps \Big\}.
\end{align*}
When the underlying problem class $\Pro$ is clear from context, we will write simply $K^{(r)}(\eps,\delta)$ and $K^{(r)}(\eps)$. Moreover, when $r=1$ we will simply write $K_\Pro(\cdot)$ or $K(\cdot)$.
\end{definition}
\noindent The following is immediate from definitions (the proof is in Appendix~\ref{app:proofs}):
\begin{proposition}\label{prop:prob_det_bounds} For any $\Pro, r \ge 1, \eps > 0, \delta \in (0,1)$,
\begin{align*}
K^{(r)}_\Pro(\eps/\delta,\delta) \le K^{(r)}_\Pro(\eps).
\end{align*}
\end{proposition}

	The complexities $K^{(r)}_{\Pro}(\eps,\delta)$ and $K^{(r)}_{\Pro}(\eps)$ capture the intrinsic difficulty of sequential optimization over the problem class $\Pro$ using any finite-step algorithm. However, most iterative optimization algorithms used in practice (such as stochastic gradient descent) are not run for a prescribed finite number of steps. Instead, they are run for however many steps are necessary until a desired accuracy is reached. Moreover, the error of the successive candidate minimizers produced by such an algorithm should decay monotonically with time. This observation motivates the following definitions:
	
\begin{definition}A {\em weak infinite-step algorithm} for $\Pro = (\cX,\Fun,\Ora)$ is a sequence of mappings $\Alg = \{\Alg_t : \cX^{t-1} \times \cY^{t-1} \to \cX \}^{\infty}_{t=1}$. The set of all weak infinite-step algorithms for $\Pro$ will be denoted by $\Algs_\infty(\Pro)$.	\end{definition}
	
	\begin{definition}\label{def:anytime} Given a problem class $\Pro = (\cX,\Fun,\Ora)$ and some $r \ge 1$, an algorithm $\Alg \in \Algs_\infty(\Pro)$ is {\em $r$-anytime} if
	\begin{align}\label{eq:anytime}
	\Err^{(r)}_\Alg(t,\Fun) \deq \sup_{f \in \Fun} \E \err^r_\Alg(t,f) \to 0 \qquad \text{as } t \to \infty.
	\end{align}
	\end{definition}

\noindent We can now ask about fundamental limits on the rate of convergence in \eqref{eq:anytime}:

	\begin{definition} For any problem class $\Pro$, we define the {\em $r$-anytime exponent} as
	\begin{align*}
	& \gamma^{(r)}_\Pro \deq \sup \Big\{ \gamma \ge 0 : \exists \Alg \in \Algs_\infty(\Pro) \nonumber\\
	& \qquad \qquad \emph{\,s.t.\,} \limsup_{t \to \infty} t^\gamma \cdot \Err^{(r)}_\Alg(t,\Fun) < \infty \Big\}
	\end{align*}
	\end{definition}

\noindent According to the above definitions, the candidate minimizer $X_{t+1}$ produced by a weak infinite-step algorithm $\Alg \in \Algs_\infty(\Pro)$ after $t$ queries is simultaneously the query at time $t+1$. Many algorithms used in practice, such as stochastic gradient descent, are weak infinite-step algorithms. A more general class of algorithms, which we may call {\em strong} infinite-step algorithms, would also include strategies in which the process of issuing queries (i.e., gathering information about the objective) is separated from the process of generating candidate minimizers. Stochastic gradient descent with trajectory averaging \cite{Polyak,NJLS09} is an example of such a strong algorithm. We do not consider strong infinite-step algorithms in this paper (except for a brief discussion in Appendix~\ref{app:finite_vs_infinite}), although their study is an interesting and important avenue for further research.

\section{Examples of problem classes and preview of selected results}
\label{sec:examples}

The following six examples show the variety of settings captured by our framework, ranging from ``standard" optimization problems to such scenarios as parameter estimation, sequential experimental design, and active learning.

\begin{example}\label{ex:lip}
\sloppypar {\em Given $L > 0$, let $\Fun^L_\Lip$ be the set of all convex functions $f: \cX \to \Reals$ that are $L$-Lipschitz, i.e., 
\begin{align*}
|f(x) - f(y) | \le L \| x - y \|, \qquad \forall x,y \in \cX.
\end{align*}
Let $\cY = \Reals \times \Reals^\Dim$ and let $P(dy|f,x)$ be a point mass concentrated at $(f(x),g(x))$, where, for each $x \in \cX$, $g(x)$ is an arbitrary subgradient in $\partial f(x)$. This oracle provides noiseless {\em first-order information}. When $L=1$, we will write $\Fun_\Lip$ instead of $\Fun^1_\Lip$.
}
\end{example}

\begin{example}\label{ex:lip_noisy} {\em Take $\Fun^L_\Lip$ as above, but now suppose that the oracle responds with
\begin{align*}
Y = (f(x) + W, g(x) + Z),
\end{align*}
where $W \in \Reals$ and $Z \in \Reals^\Dim$ are zero-mean random variables with finite second moments. Thus, any algorithm receives {\em noisy} first-order information, and the oracle is local.}
\end{example}

\begin{example}\label{ex:sc}
{\em Given $\kappa > 0$, let $\Fun^\kappa_\SC$ be the set of all differentiable functions $f : \cX \to \Reals$ that are $\kappa$-\emph{strongly convex}, i.e.,
\begin{align*}
f(x) \ge f(y) + \nabla f(y)^\tr (x-y) + \frac{\kappa^2}{2} \| x- y \|^2, \,\, \forall x, y \in \cX.
\end{align*}
As in the previous example, the oracle responds with
\begin{align*}
Y = (f(x) + W, g(x) + Z),
\end{align*}
where $W \in \Reals$ and $Z \in \Reals^\Dim$ are zero-mean random variables with finite second moments. When $\kappa=1$, we will write $\Fun_\SC$ instead of $\Fun^\kappa_\SC$.}
\end{example}

\begin{example}\label{ex:statistical} {\em Fix a compact convex set $\cX \subset \Reals^\Dim$ and a family of probability measures $\{ P_\theta : \theta \in \cX \}$ on $(\cY,\cB_\cY)$. Consider the class of convex functions
\begin{align}\label{eq:loss}
\Fun = \left\{ f_\theta(x) : \theta \in \cX \right\},
\end{align}
such that for every $\theta \in \cX$ $f_\theta(\theta) = \min_{x \in \cX} f_\theta(x)$. Consider also the oracle $\Ora = (\cY,P)$, defined by
\begin{align}\label{eq:stat_oracle}
P(dy|f_\theta,x) = P_\theta(dy), \qquad \forall (\theta,x) \in \cX \times \cX.
\end{align}
This oracle ignores the query $x$ and simply outputs a random element $Y \sim P_\theta$. The problem class $(\cX,\Fun,\Ora)$ thus describes the statistical problem of \emph{estimating the parameter of a probability distribution}. More generally, we can consider the function class
\begin{align}\label{eq:contrast}
\Fun = \left\{ f_\theta(\cdot) = \E_\theta [F(\cdot,Y)] \equiv \int_\cY F(\cdot,y) P_\theta(dy) : \theta \in \cX \right\},
\end{align}
where we assume that:
\begin{itemize}
\item For each fixed $y \in \cY$, the function $x \mapsto F(x,y)$ is convex
\item $f_\theta(\theta) = \min_{x \in \cX} f_\theta(x)$
\end{itemize}
The second condition says that $F : \cX \times \cY \to \Reals$ is a {\em contrast function} \cite{MCE}. Most classical problems in statistical inference, such as estimating the mean, the median, or the variance of a distribution, can be cast as minimizing a convex contrast function of the form \eqref{eq:contrast}. For instance, if $\cX \subset \Reals$, $\cY = \Reals$, $\E_\theta[Y] = \theta$ for each $\theta \in \cX$, and $F(x,y) = (x-y)^2$, then
$$
f_\theta(x) = \E_\theta[(Y-x)^2]
$$
with $f_\theta(x) \ge f_\theta(\theta) \equiv \var(P_\theta)$, so we recover the problem of estimating the mean.
}
\end{example}

\begin{example}\label{ex:design}{\em As we have just seen, the queries are of no use in statistical estimation since the samples the statistician obtains depend only on the unknown parameter $\theta$. By contrast, the setting in which the statistician's queries {\em do} affect the observations is known as {\em sequential experimental design} \cite{Fed72,MacKay92,Pan05}. Consider the case when $\cX \subset \Reals^\Dim$ is compact and convex, as in the above example. Suppose also that we have two families of probability measures on $\cY$, $\{Q_\theta : \theta \in \cX\}$ and $\{ P_{\theta,x} : (\theta,x) \in \cX \times \cX\}$. The function class is as in \eqref{eq:contrast} but with $Q_\theta$ replacing $P_\theta$, while the oracle now is defined by
\begin{align*}
	P(dy|f_\theta,x) = P_{\theta,x}(dy), \qquad \forall (\theta,x) \in \cX \times \cX.
\end{align*}
Thus, the role of $Q_\theta$ is to provide a measure of performance (or goodness-of-fit) of the final estimate of $\theta$, while $P_{\theta,x}$ describes the experimental model (i.e., the relationship between the input $X$ and the response $Y$ given the parameter $\theta$).
}
\end{example}

\begin{example}\label{ex:active} {\em Our last example is at the intersection of statistical learning theory and sequential experimental design. Let
\begin{align*}
\cX = [0,1], \,\, \cY = \{-1,+1\}, \,\, \Fun = \{ f_\theta(x) = |x-\theta| : \theta \in \cX\}.
\end{align*}
To define the oracle, suppose that there exist some $0 < c,C < 1/2$ and $\kappa \in [1,\infty)$, such that
\begin{align*}
c | x - \theta |^{\kappa - 1} \le \left| P(Y=1|f_\theta,x) - 1/2 \right| \le C | x - \theta|^{\kappa - 1},
\end{align*}
where the first inequality holds for all $x$ in a sufficiently small neighborhood of $\theta$. This oracle provides a noisy subgradient of $f_\theta$ at $x$, and the amount of noise depends on the distance between $x$ and $\theta$. This problem class is related to {\em active learning} of a threshold function on the unit interval \cite{castro2008minimax}, and will be treated in detail in Section~\ref{sec:anytime}.}
\end{example}

We now briefly discuss some of the lower bounds that arise from the techniques introduced in the paper. First, Theorem~\ref{thm:general_bound} in Section~\ref{sec:lower_bounds} implies a general lower bound of the form 
$$
\Omega\big(n^{\alpha}\log \left(1/\eps\right)\big)
$$
on the number of oracle calls required to $\eps$-minimize every function in a given class, where the exponent $\alpha > 0$ depends on the geometry of the problem domain $\cX$ and on the complexity of the instance space $\Fun$.  For convex Lipschitz functions and noiseless first-order oracles (Example~\ref{ex:lip}), or more generally for stochastic oracles that are sufficiently ``informative'' in a sense we make precise, this lower bound holds with $\alpha = 1$ (cf.~the discussion right after Theorem~\ref{thm:general_bound}). This lower bound is known to be optimal in the noiseless case \cite{NemYud83} and in certain noisy scenarios when $n=1$ \cite{BZ}; however, our techniques lead to a much more transparent proof of the bound.

For the noisy first-order oracle with zero-mean Gaussian noise of variance $\sigma^2$, we obtain lower bounds of the form
$$
\Omega\big(\sigma^2 n^{\alpha_1}(1/\eps)^{\alpha_2}\big),
$$
where the exponent $\alpha_1$ depends, as before, on the geometry of $\cX$, on the complexity of $\Fun$, as well as on whether the oracle supplies full first-order information (function value and subgradient) or just the subgradient. The exponent $\alpha_2$ depends on the details of the function class $\Fun$. More specifically:
\begin{itemize}
	\item for $\Fun_\Lip$ (Example~\ref{ex:lip_noisy}), we have $\alpha_2 = 2$ (Theorem~\ref{thm:Lip_bounds} in Section~\ref{sec:lower_bounds});
	\item for $\Fun_\SC$ (Example~\ref{ex:sc}), we have $\alpha_2 = 1$ (Theorem~\ref{thm:SC_bounds} in Section~\ref{sec:lower_bounds}).
\end{itemize}
The corresponding result for convex Lipschitz functions in $n=1$ can be found in \cite{NemYud83,ShaNem05}, yet we obtain the optimal dependence on $n$ for higher dimensions. Our lower bound for strongly convex functions seems to be new; in particular, Nemirovski and Yudin \cite{NemYud83} only consider the noiseless case, while Agarwal {\em et al.~}\cite{Aga09,Aga10} consider noisy first-order oracles, but with a different oracle model, which does not allow additive noise due to a coin-tossing construction. Ignoring the dependence on the dimension, we also obtain the error decay rate $\Omega(\sigma^2/t)$ for $\Fun_\SC$ when we restrict ourselves to anytime infinite-step algorithms (Theorem~\ref{thm:anytime_SC} in Section~\ref{sec:anytime}). To the best of our knowledge, such analysis does not appear anywhere else in the literature. The bounds of Eq.~(\ref{eq:IBC_bound}) essentially capture the fundamental limits of strongly convex programming in one dimension and can be easily deduced using our techniques (a sketch of the derivation is given in Section~\ref{ssec:illustration}). We also derive new (and tighter) lower bounds on anytime algorithms for minimizing higher-order polynomials under a second-moment error criterion (Theorems~\ref{thm:poly_arbitrary} and \ref{thm:poly_anytime} in Section~\ref{sec:anytime}).

Apart from ``standard'' optimization problems, our framework seamlessly captures several statistical problems with an optimization flavor. In particular, in Section~\ref{ssec:stats} we look at information-based complexity of statistical estimation and sequential experimental design (Examples~\ref{ex:statistical} and \ref{ex:design}, respectively). Here we do not aim at obtaining tight rates for specific settings of interest, but rather show the connections to the techniques employed in statistics. Finally, we show in Section~\ref{ssec:active} that our methodology leads to a particularly easy derivation of a lower bound for the active learning problem of Example~\ref{ex:active}. This bound was previously obtained in \cite{castro2008minimax} using a much more involved argument relying on a careful construction of a ``difficult" subset of functions.

Overall, our main contributions are the development of a general framework that captures many diverse settings with optimization flavor, as well as a novel analysis that takes into account the effect of feedback upon the dynamics of the interaction between the algorithm and the oracle.

\section{Setting the stage: optimization vs.\ hypothesis testing with feedback}
\label{sec:HTF}

We now lay down the foundations of our information-theoretic method for determining lower bounds on the information complexity of convex programming. The basic strategy is to show that the minimum number of oracle queries is constrained by the average rate at which each new query can reduce the algorithm's uncertainty about the function being optimized. 

Conceptually, our techniques are akin to the ones used in statistical literature to obtain minimax lower bounds on the risks of estimation procedures \cite{Yu97,YanBar99,Tsy09}. The main idea is this. Given a problem class $\Pro = (\cX,\Fun,\Ora)$, we construct a ``difficult'' finite subclass $\Fun' = \{f_0,\ldots,f_{N-1}\} \subset \Fun$, such that the functions in it are nearly indistinguishable from one another based on the information supplied by the oracle in response to any possible query, and yet they are sufficiently far apart from one another, so that a candidate approximate minimizer for any one of them fails to minimize all the remaining functions to the same accuracy. Once such a class is constructed, we consider a fictitious situation in which Nature selects an element of $\Fun'$ {\em uniformly at random}. Then for every $T$-step algorithm $\Alg \in \Algs_T(\Pro)$ we can construct a probability space $(\Omega,\cB,\Prob)$ with the following random variables defined on it:
\begin{itemize}
	\item $M \in \{0,\ldots,N-1\}$, which encodes the random choice of a problem instance in $\Fun'$
	\item $X^{T+1} \in \cX^{T+1}$, where $X^T$ are the queries issued by $\Alg$ and $X_{T+1}$ is the candidate minimizer
	\item $Y^T \in \cY^T$ are the responses of $\Ora$ to the queries issued by $\Alg$.
\end{itemize}
These variables describe the interaction between Nature, the algorithm, and the oracle, and thus have the causal ordering
\begin{align*}
	M,X_1,Y_1,\ldots,X_t,Y_t,\ldots,X_T,Y_T,X_{T+1},
\end{align*}
where, $\Prob$-almost surely,
\begin{align}
	&\Prob(M=i) = \frac{1}{N} \label{eq:Nature_choice} \\
	&\Prob(X_t \in A | M,X^{t-1},Y^{t-1}) = \ind{\Alg_t(X^{t-1},Y^{t-1}) \in A} \nonumber\\
	&\Prob(Y_t \in B | M, X^t,Y^{t-1}) = P(B|f_M,X_t),  \nonumber
\end{align}
for all $i \in \{0,\ldots,N-1\}, A \in \cB_\cX, B \in \cB_\cY$. In other words, $M \to (X^{t-1},Y^{t-1}) \to X_t$ and $(X^{t-1},Y^{t-1}) \to (M,X_t) \to Y_t$ are Markov chains for every $t$. 

The reason for such punctilious bookkeeping is that now we can relate the problem faced by $\Alg$ to sequential hypothesis testing with feedback, as defined by Burnashev \cite{Burnashev}. We can think of $M$ as encoding the choice of one of $N$ equiprobable hypotheses. At each time $t$, the algorithm issues a query $X_t$ and receives an observation $Y_t$ which is stochastically related to $X_t$ and $M$ via the kernel $\Prob(dY_t|M,X_t) = P(dY_t|f_M,X_t)$. The current query may depend only on the past queries and observations. At time $T+1$, the algorithm produces a candidate minimizer, $X_{T+1}$. As we will shortly demonstrate, we can use the information available to $\Alg$ at time $T+1$ to construct an {\em estimate} $\wh{M}_T$ of the true hypothesis $M$.\footnote{It is important to keep in mind that the hypothesis testing set-up is purely fictitious --- indeed, $\Alg$ may or may not know that the problem instances are drawn at random among $\{f_0,\ldots,f_{N-1}\}$, rather than arbitrarily from the entire instance space $\Fun$. The point is, though, that the average performance of $\Alg$ on $\Fun'$ cannot be better than its worst-case performance on $\Fun$. In statistical terms, the minimax risk of $\Alg$ over $\Fun$ is bounded below by the Bayes risk over any subset of $\Fun$.} Once this is done, we can analyze the {\em mutual information} $I(M; \wh{M}_T)$, which is well-defined because we have specified $\Prob$. In particular, the analysis hinges on the following observations. Suppose that $\Alg$ is such that for some $r \ge 1$, $\eps > 0$, and $\delta \in (0,1)$ we have
\begin{align*}
	\Pr\Big( \err^r_\Alg(T,f) \ge \eps\Big) \le \delta, \qquad \forall f \in \Fun
\end{align*}
where the probability is w.r.t.\ the randomness in the oracle's responses. Then, first of all,
\begin{align*}
	\Prob \Big(\err^r_\Alg(T,f_M) \ge \eps \Big) \le \delta.
\end{align*}
We will use this fact, together with the ``geometric'' distinguishability of the functions $\{f_i\}$, to show that $\Prob(\wh{M}_T \neq M) \le \delta$ and, as a consequence, that there exists some $\Psi_1(r,\eps,\delta) > 0$, such that
\begin{align}\label{eq:Psi_lower}
	I(M; \wh{M}_T) \ge \Psi_1(r,\eps,\delta).
\end{align}
In other words, a good algorithm should be able to obtain a nontrivial amount of information about the hypothesis $M$. On the other hand, by the data processing inequality, $I(M; \wh{M}_T) \le I(M; X^T,Y^T)$, and we will use statistical indistinguishability of $\{f_i\}$, as well as the structure of the oracle, to obtain an upper bound of the form
\begin{align}\label{eq:Psi_upper}
	I(M; \wh{M}_T) \le T \Psi_2(r,\eps)
\end{align}
with some $\Psi_2(r,\eps) < +\infty$. The two bounds are then combined to yield
\begin{align}\label{eq:Psi_bound}
	T \ge \frac{\Psi_1(r,\eps,\delta)}{\Psi_2(r,\eps)} \quad \Longrightarrow \quad K^{(r)}_\Pro(\eps,\delta) \ge \frac{\Psi_1(r,\eps,\delta)}{\Psi_2(r,\eps)}.
\end{align}

\subsection{An illustrative example}
\label{ssec:illustration}

To illustrate our method in action, we will sketch the derivation of the nontrivial part of the lower bound in \eqref{eq:IBC_bound}, i.e., when $\eps \in (0,1/2)$. Let
\begin{align*}
	x^*_0 &= \begin{cases}
	1/2 - \sqrt{2\eps}, & \eps \in (0,1/8)\\
	0, & \eps \in [1/8,1/2)
\end{cases} \\
	x^*_1 &= \begin{cases}
	1/2 + \sqrt{2\eps}, & \eps \in (0,1/8)\\
	1, & \eps \in [1/8,1/2).
\end{cases}
\end{align*}
It is easy to see that $x^*_0,x^*_1 \in [0,1]$. Consider two functions
$$
f_m(x) = \frac{1}{2}(x-x^*_m)^2, \qquad m \in \{0,1\}.
$$
A simple calculation shows that for any $x \in [0,1]$ such that
$$
f_0(x) - f^*_0 = f_0(x) = \frac{1}{2}(x-x^*_0)^2 < \eps
$$
we must have $f_1(x) - f^*_1 = f_1(x) > \eps$, and the same holds with the roles of $f_0$ and $f_1$ reversed. Thus, any $\eps$-minimizer of $f_0$ fails to $\eps$-minimize $f_1$, and vice versa.

On the other hand, the probability distribution of the output of the first-order Gaussian oracle \eqref{eq:simple_Gaussian_oracle} for any query $x \in [0,1]$ when $M=0$ is very close to its $M=1$ counterpart. Indeed, letting $Y \in \Reals^2$ denote the output of the oracle, we have
$$
\Prob_{Y|M=m,X=x} = \Normal\left(\frac{1}{2}(x-x^*_m)^2, \sigma^2\right) \otimes \Normal(x-x^*_m, \sigma^2).
$$
Then it is not hard to show that, for $m \in \{0,1\}$,
\begin{align}\label{eq:KL_bound}
D(\Prob_{Y|M=m,X=x} \| \Prob_{Y|M=1-m,X=x}) = O\left( \frac{\eps}{\sigma^2}\right)
\end{align}
In other words, the functions $f_0$ and $f_1$ are nearly indistinguishable from one another based on the outcome of a single query.

Now suppose that Nature selects an index $M \in \{0,1\}$ uniformly at random. Consider a $T$-step algorithm that $\eps$-minimizes every function in the class $\Fun$ defined in \eqref{eq:simple_function_class} with probability at least $1-\delta$, where $\delta \in (0,1/2)$. Let $r=1$. Then Lemma~\ref{lm:fano} in Section~\ref{ssec:info_bounds} can be used to show that the lower bound \eqref{eq:Psi_lower} holds with
$$
\Psi_1(1,\eps,\delta) = \log 2 - h_2(\delta) > 0,
$$
where $h_2(\delta) \deq - \delta \log \delta - (1-\delta) \log (1-\delta)$ is the binary entropy function. On the other hand, using Lemmas~\ref{lm:info_bounds} and \ref{lm:refined_info_bounds}, as well as Eq.~\eqref{eq:KL_bound}, we can show that the upper bound \eqref{eq:Psi_upper} holds with
\begin{align*}
\Psi_2(1,\eps) &= \max_{m \in \{0,1\}} \max_{x \in [0,1]} D(\Prob_{Y|M=m,X=x} \| \Prob_{Y|M=1-m,X=x}) \\
&= O\left( \frac{\eps}{\sigma^2}\right).
\end{align*}
Hence, according to \eqref{eq:Psi_bound}, any $T$-step algorithm that $\eps$-minimizes every function in the class $\Fun$ of Eq.~(\ref{eq:simple_function_class}) with probability at least $1-\delta$ must satisfy
$$
T = \Omega\left( \frac{\sigma^2 (\log 2 - h_2(\delta))}{\eps} \right).
$$
From this and from Proposition~\ref{prop:prob_det_bounds}, we can obtain the lower bound $\Omega(\sigma^2/\eps)$ of Eq.~(\ref{eq:IBC_bound}). The matching upper bound $O(\sigma^2/\eps)$ is achieved by stochastic gradient descent \cite{NJLS09}.

\subsection{Reduction to hypothesis testing with feedback}
\label{ssec:reduction}

We now develop our information-theoretic methodology in the general setting of Section~\ref{ssec:formal_defs}.

Let us fix a problem class $\Pro = (\cX,\Fun,\Ora)$. To set up our analysis, we first endow the instance space $\Fun$ with a ``distance'' $\disc(\cdot,\cdot)$ that has the following property: for any $x \in \cX$ and any $\eps > 0$,
\begin{align}\label{eq:pseudometric}
\disc(f,g) \ge 2\eps \text{ and } f(x) < f^* + \eps \,\, \Longrightarrow \,\, g(x) > g^* + \eps.
\end{align}
In other words, an $\eps$-minimizer of a function cannot simultaneously be an $\eps$-minimizer of a distant function. It is easy to construct a $\disc$ satisfying \eqref{eq:pseudometric} for any particular class $\Fun$ of continuous functions, although such a $d$ need not be a metric. For example, if we consider the class
\begin{align*}
 \Fun_\Theta \deq \{ f_\theta(x) = \| x - \theta \| : \theta \in \Theta \}
\end{align*}
for some $\Theta \subset \cX$, then $\disc(f_\theta,f_{\theta'}) = \| \theta - \theta' \|$ satisfies (\ref{eq:pseudometric}).  Indeed, $\| \theta - \theta' \| \ge 2\eps$ and $\| x - \theta \| < \eps$ imply $\| x - \theta' \| > \eps$ by the triangle inequality. 
 For a general $\Fun$, we can also define
\begin{align*}
 \disc(f,g) \deq \inf_{x \in \cX} [f(x) + g(x)] - [f^* + g^*],
\end{align*}
the distance-like function introduced in \cite{Aga09,Aga10}. This definition coincides with $\disc(f_\theta,f_{\theta'}) = \| \theta - \theta' \|$ for the parametric set $\Fun_\Theta$; however, (\ref{eq:pseudometric}) is the most general requirement. Note that we will often implicitly restrict our consideration to a {\em subclass} of $\Fun$ and define an appropriate $\disc$ on that subclass.
 
Let us fix the exponent $r \ge 1$ and consider any finite $\Fun' = \{f_0,\ldots,f_{N-1}\} \subset \Fun$, such that any two distinct $f_i,f_j \in \Fun'$ are at least $2\eps^{1/r}$ apart in $\disc(\cdot,\cdot)$. Given any $T \in \Naturals$ and an algorithm $\Alg \in \Algs_T(\Pro)$, we can now construct the probability space $(\Omega,\cB,\Prob)$, as described in the introduction to this section. Given $X_{T+1}$, the output of $\Alg$, we can define the ``estimator"
\begin{align}\label{eq:canonical_estimator}
\wh{M}_T(X^T,Y^T)  \deq \argmin_{i=0,\ldots,N-1} [f_i(X_{T+1}) - f^*_i],
\end{align}
which simply selects that function in $\Fun'$ for which the error of $X_{T+1}$ is the smallest. Since $X_{T+1}$ is $\sigma(X^T,Y^T)$-measurable, the estimator $\wh{M}_T$ is indeed a function only of the information available to $\Alg$ after time $T$.

\subsection{Information bounds}
\label{ssec:info_bounds}

The main object of interest will be the mutual information $I(M; \wh{M}_T)$. We first show that any ``good" $T$-step algorithm obtains a nonzero amount of information about $M$ at the end of its operation:

\begin{lemma}\label{lm:fano} Fix some $r \ge 1$, $\delta \in (0,1/2)$, and $\eps > 0$. Suppose $\Alg \in \Algs_T(\Pro)$ attains
\begin{align}\label{eq:good_alg}
\sup_{f \in \Fun} \Pr \Big( \err^r_\Alg(T,f) \ge \eps \Big) \le \delta.
\end{align}
Let $\Fun' \subset \Fun$ be a finite set $\{f_0,\ldots,f_{N-1}\}$ of functions, such that
\begin{align*}
	\disc(f_i,f_j) \ge 2\eps^{1/r}, \qquad \forall i \neq j.
\end{align*}
Let $M$ be uniformly distributed on $\{0,1,\ldots,N-1\}$, and suppose that $\Alg$ is fed with the random problem instance $f_M \in \Fun'$. If $N > 4$, then the estimator $\wh{M}_T$ defined in (\ref{eq:canonical_estimator}) satisfies the bound
\begin{align}\label{eq:fano}
I(M; \wh{M}_T) \ge (1-\delta) \log N - \log 2 > 0.
\end{align}
If $N = 2$, then
\begin{align}\label{eq:binary_fano}
I(M; \wh{M}_T) \ge \log 2 - h_2(\delta) > 0,
\end{align}
where $h_2(\delta) \deq -\delta \log \delta - (1-\delta) \log (1-\delta)$ is the binary entropy function.
\end{lemma}

\begin{remark} {\em In the sequel, we will consider only the cases when the set $\Fun'$ is either ``rich", so that $N \gg 4$, or has only two elements, so $N = 2$.}
\end{remark}

\begin{proof} Consider an algorithm $\Alg$ with the claimed properties. Define, for each $i$, the event
\begin{align*}
	E_i \deq \left\{ \err^r_\Alg(T,f_i) \ge \eps \right\}.
\end{align*}
We first show that the event $\{ \wh{M}_T \neq i\}$ implies $E_i$. Indeed, if $E_i$ does not occur, then from the fact that $d(f_i,f_j) \ge 2\eps^{1/r}$ for all $j \neq i$ and from \eqref{eq:pseudometric} we deduce that
\begin{align*}
f_{j}(X_{T+1}) - f^*_{j} > \eps^{1/r} > f_i(X_{T+1}) - f^*_i, \qquad \forall j \neq i
\end{align*}
so it must be the case that $\wh{M}_T = i$. Therefore,
\begin{align*}
\delta 
&\ge \max_{i = 0,\ldots,N-1} \Prob(E_i | M=i) \\
&\ge \max_{i = 0,\ldots,N-1} \Prob(\wh{M}_T \neq i|M=i) \\
&\ge \Prob(\wh{M}_T \neq M).
\end{align*}
Now suppose that $N > 4$. Then we can invoke the following version of Fano's inequality \cite{HanVer94}:
\begin{align*}
\Prob(\wh{M}_T \neq M) \ge 1 - \frac{I(M; \wh{M}_T) + \log 2}{\log N}.
\end{align*}
Rearranging, we get (\ref{eq:fano}). When $N = 2$, we use a stronger form of Fano's inequality (see, e.g.,~Section~2.10 in \cite{CovTho06}):
\begin{align*}
h_2\big(\Prob(\wh{M}_T \neq M)\big) \ge \log 2 - I(M; \wh{M}_T).
\end{align*}
Since  $\delta \mapsto h_2(\delta)$ is monotone increasing on $[0,1/2]$, we get $h_2(\delta) \ge \log 2 - I(M; \wh{M}_T)$. Rearranging, we get (\ref{eq:binary_fano}).
\end{proof}

On the other hand, the amount of information $I(M; \wh{M}_T)$ cannot be too large:

\begin{lemma}\label{lm:info_bounds} Any estimator $\wh{M} : \cX^T \times \cY^T \to \{0,\ldots,N-1\}$ [and, in particular, the estimator $\wh{M}_T$ defined in (\ref{eq:canonical_estimator})] satisfies
\begin{align}\label{eq:info_upper_bound}
I(M; \wh{M}) \le \sum^T_{t=1} I(M; Y_t|X^t,Y^{t-1}).
\end{align}
\end{lemma}

\begin{remark} {\em The terms $I(M; Y_t | X^t, Y^{t-1})$ have analogues in the literature on information-theoretic experimental design (see, e.g.,~\cite{Fed72,MacKay92}). In that context, they represent the average reduction of uncertainty about the unknown variable $M$ after observing the {\em experimental outcome} $Y_t$ based on the {\em design point} $X_t = \Alg_t(X^{t-1},Y^{t-1})$. 
}
\end{remark}

\begin{proof} We have
\begin{align}
& I(M; \wh{M}) \le I(M; X^T,Y^T) \label{eq:step1} \\
&= \sum^T_{t=1} I(M; X_t,Y_t | X^{t-1},Y^{t-1}) \label{eq:step2} \\
&= \sum^T_{t=1} [I(M; X_t | X^{t-1},Y^{t-1}) + I(M; Y_t | X^t, Y^{t-1})] \label{eq:step3} \\
&= \sum^T_{t=1} I(M; Y_t | X^t, Y^{t-1}), \label{eq:step4}
\end{align}
where (\ref{eq:step1}) is a consequence of the data processing inequality; (\ref{eq:step2}) and (\ref{eq:step3}) use the chain rule; and (\ref{eq:step4}) uses the fact that $M \to (X^{t-1},Y^{t-1}) \to X_t$ is a Markov chain. 
\end{proof}

\subsection{Refinement of the upper bounds}

Lemmas~\ref{lm:fano} and \ref{lm:info_bounds} are the two main elements of our approach. In order to apply them, we need to get a handle on the conditional mutual information terms on the right-hand side of \eqref{eq:info_upper_bound}. The following two lemmas, whose proofs can be found in Appendix~\ref{app:proofs}, give us just the right tools for that:

\begin{lemma}\label{lm:oblivious_info_bounds} Consider any estimator $\wh{M} : \cX^T \times \cY^T \to \{0,\ldots,N-1\}$. Then, considering any realization of the oracle $\Ora$ in the form \eqref{eq:SLO}, we have the bound
\begin{align*}
I(M; \wh{M}) \le \sum^T_{t=1} I(U_t; Y_t),
\end{align*}
where $U_t = \psi(f_M,X_t)$ is the output of the ``deterministic part'' of the oracle.
\end{lemma}

\begin{lemma}\label{lm:refined_info_bounds} Consider any estimator $\wh{M} : \cX^T \times \cY^T \to \{0,\ldots,N-1\}$. For any sequence of conditional probability measures $\{ \ProbQ_{Y_t|X^t,Y^{t-1}} \}^T_{t=1}$ on $(\Omega,\cB)$ satisfying the conditions
\begin{align}\label{eq:AC_condition}
\Prob_{Y_t|X^t,Y^{t-1}} \ll \ProbQ_{Y_t|X^t,Y^{t-1}}, \qquad t=1,\ldots,T
\end{align}
we have the bound
\begin{align}\label{eq:refined_info_bound}
I(M; \wh{M}) \le \sum^T_{t=1} D\big(\Prob_{Y_t|M,X^t,Y^{t-1}} \big\| \ProbQ_{Y_t | X^t,Y^{t-1}} \big| \Prob_{M,X^t,Y^{t-1}} \big)
\end{align}
\end{lemma}

\begin{remark} {\em By hypothesis on the behavior of the oracle, $(X^{t-1},Y^{t-1}) \to (M,X_t) \to Y_t$ is a Markov chain. Hence, $\Prob_{Y_t|M,X^t,Y^{t-1}}$ in \eqref{eq:refined_info_bound} can be replaced with $\Prob_{Y_t|M,X_t}$.}
\end{remark}

The key to using Lemma~\ref{lm:refined_info_bounds} is in the judicious choice of the ``comparison" measures $\ProbQ_{Y_t|X^t,Y^{t-1}}$. In particular, we will use two different strategies of choosing the $\ProbQ$'s, which in turn lead to two different types of bounds:
\begin{itemize}
\item {\bf Information Radius (IR) bound ---} This bound is useful for analyzing arbitrary finite-time algorithms. For each $t$, take $\ProbQ_{Y_t|X^t,Y^{t-1}} = \ProbQ_{Y_t|X_t}$ to be the mixture
\begin{align*}
\ProbQ_{Y_t|X_t} = \frac{1}{N}\sum^{N-1}_{i=0} \Prob_{Y_t|M=i,X_t}.
\end{align*}
Then, letting $M'$ denote an independent copy of $M$ and noting that $\ProbQ_{Y_t|X_t} = \E_{M'} \Prob_{Y_t|M',X_t}$, we obtain
\begin{align}
& I(M; \wh{M}) \nonumber\\
&= \sum^T_{t=1} D(\Prob_{Y_t|M,X_t} \| \E_{M'} \Prob_{Y_t|M',X_t} | \Prob_{M,X_t})  \nonumber \\
& \le \sum^T_{t=1} \E_{M'} D(\Prob_{Y_t|M,X_t} \| \Prob_{Y_t|M',X_t} | \Prob_{M,X_t}) \label{eq:Jens}\\
& = \sum^T_{t=1} \E_{M,X_t} \E_{M'} D(\Prob_{Y_t|M,X_t} \| \Prob_{Y_t|M',X_t}) \nonumber \\
& \le T \max_{i,j} \sup_{x \in \cX} D(\Prob_{Y|M=i,X=x} \| \Prob_{Y|M=j,X=x}),\label{eq:IR_bound}
\end{align}
where \eqref{eq:Jens} follows from Jensen's inequality and convexity of the divergence. The use of the term ``information radius" is inspired by an analogous concept in the theory of information-based complexity \cite{Plaskota}:  the divergence $D(\Prob_{Y|M=i,X=x} \| \Prob_{Y|M=j,X=x})$ quantifies how close, in a statistical sense, the oracle's responses are for a given query point $x \in \cX$ and a given pair $i,j$. Viewing the random variable $Y \sim \Prob_{Y|M=i,X=x}$ as (stochastic, noisy) {\em information} about the function $f_i$ at the point $x$, we can interpret the quantity multiplying $T$ in \eqref{eq:IR_bound} as a measure of {\em ambiguity} of this information. We use IR bounds in Section~\ref{sec:lower_bounds}.
\item {\bf Lyapunov Function (LF) bound ---} This bound is useful for analyzing anytime algorithms. It relies on the idea that, with certain types of problem classes, the oracle responds with ``pure noise" whenever the query point happens to hit upon a minimizer. In other words, there exists a probability measure $Q^*$ on $\cY$, such that
\begin{align*}
P(dy|f,x) = Q^*(dy) \qquad \text{if } x \in \argmin_\cX f,
\end{align*}
where $\argmin_\cX f \deq \{ x \in \cX: f(x) = f^*\}$. Moreover, for an anytime algorithm it is often the case that the conditional divergence $D(\Prob_{Y_t|M,X_t} \| \ProbQ^*_{Y_t}|\Prob_{X_t})$, where $\ProbQ^*_{Y_t}$ is an independent copy of $Q^*$, decreases with $t$, and hence can be thought of as a Lyapunov function for the problem at hand (in fact, Lyapunov functions of the divergence type have been used before to analyze the convergence of specific stochastic optimization algorithms \cite{Vilk}). This leads to the natural choice of $\ProbQ_{Y_t|X^t,Y^{t-1}} = \ProbQ^*_{Y_t}$, and to the bound
\begin{align}
I(M; \wh{M}) \le \sum^T_{t=1} D(\Prob_{Y_t|M,X_t} \| \ProbQ^*_{Y_t} | \Prob_{M,X_t}).
\label{eq:LF_bound}
\end{align}
We apply the LF bounds in Section~\ref{sec:anytime} to the study of anytime algorithms.
\end{itemize}
We should point out that the use of an auxiliary measure $\ProbQ$ has been pioneered by Yang and Barron \cite{YanBar99} (with later refinements by Yang \cite{YangAISTATS}) in the usual setting of statistical estimation from i.i.d.\ data; there, the relevant bound was of the form
\begin{align*}
I(M ; Y^T) &\le D(\Prob_{Y^T|M} \| \ProbQ_{Y^T} | \Prob_M) \\
&\le \max_m D(\Prob_{Y^T|M=m} \| \ProbQ_{Y^T})
\end{align*}
(cf.~\cite[p.~1571]{YanBar99}). Similarly, the ``symmetrization trick'' involving an independent copy of $M$ and an application of Jensen's inequality, as in Eq.~\eqref{eq:Jens} above, is used often in the statistics literature (cf.~\cite{Yu97} and references therein). Our innovation consists in first performing a sequential decomposition of the mutual information $I(M; X^T,Y^T)$, carefully taking into account all the Markov structures that arise due to the causality constraints that must be obeyed by the algorithm, and then choosing an appropriate auxiliary measure for each time $t = 1,\ldots,T$.

\section{Lower bounds for arbitrary algorithms}
\label{sec:lower_bounds}

We now apply the lemmas of the preceding section to the problem of deriving lower bounds on the information-based complexity of several problem classes. These bounds hold for {\em arbitrary} finite-step or infinite-step algorithms.

\subsection{A general information-theoretic lower bound}
\label{ssec:general_bound}

Our first bound applies to {\em any} problem class. However, this generality comes at a price: the bound is nontrivial (i.e.,~tight) only in certain cases. 

\begin{theorem}\label{thm:general_bound} Consider a problem class $\Pro = (\cX,\Fun,\Ora)$, with any realization of the oracle $\Ora$ in the form \eqref{eq:SLO}. Given any $\eta > 0$, define the {\em packing number}
\begin{align*}
& N(\Fun,\disc,\eta) \deq \max\Big\{ N \ge 1: \exists f_0,\ldots,f_{N-1} \in \Fun \nonumber\\
& \qquad \qquad \qquad \emph{\,\,s.t.\,} \disc(f_i,f_j) \ge 2\eta, \forall i \neq j \Big\}.
\end{align*}
Then, for any $r \ge 1$, any $\eps$ such that $N(\Fun,\disc,\eps^{1/r}) > 4$, and any $\delta \in (0,1/2)$, the following bounds hold:
\begin{align}\label{eq:log_bound_1}
K^{(r)}_{\Pro}(\eps,\delta) &\ge \frac{1}{C^*} \left[(1-\delta) \log N(\Fun,\disc,\eps^{1/r}) - \log 2 \right] \\
\label{eq:log_bound_2}
K^{(r)}_{\Pro}(\eps) &\ge \frac{1}{C^*} \left[ \frac{2}{3}\log N(\Fun,\disc,(3\eps)^{1/r}) - \log 2\right]
\end{align}
with
\begin{align*}
C^* \deq \sup_{U \in \cU_{\cX,\Fun}} I(U; Y),
\end{align*}
where the supremum is over all random variables $U$ taking values in $\cU_{\cX,\Fun} = \psi(\Fun,\cX)$, and the mutual information is between $U$ and $Y = \Phi(U,W)$, cf.~Eq.~\eqref{eq:SLO}.
\end{theorem}

\begin{remark} {\em The number $C^*$ is the {\em Shannon capacity} of the random transformation $Y = \Phi(U,W)$ when its input is constrained to lie in the information space of the deterministic oracle $\psi$. When $C^* = \infty$, the bounds \eqref{eq:log_bound_1} and \eqref{eq:log_bound_2} simply say that $K^{(r)}_\Pro(\eps,\delta) \ge 0$ and $K^{(r)}_\Pro(\eps) \ge 0$.}
\end{remark}

\begin{proof} Let $\Fun^{(r)}_\eps = \{f_0,\ldots,f_{N-1}\} \subset \Fun$, $N = N(\Fun,\disc,\eps^{1/r})$, be a maximal packing set in $\Fun$.  Given $\delta \in (0,1/2)$, consider any $T$ and any algorithm $\Alg \in \Algs_T(\Pro)$ such that $\Pr(\err^r_\Alg(T,f) \ge \eps) \le \delta$. Then we can apply Lemma~\ref{lm:fano} to get
\begin{align*}
I(M; \wh{M}_T) \ge (1-\delta) \log N - \log 2.
\end{align*}
On the other hand, from Lemma~\ref{lm:oblivious_info_bounds} and the definition of $C^*$,
\begin{align*}
I(M; \wh{M}_T) \le \sum^T_{t=1} I(U_t; Y_t) \le T C^*.
\end{align*}
Combining these two bounds, we get \eqref{eq:log_bound_1}, while \eqref{eq:log_bound_2} follows after applying Proposition~\ref{prop:prob_det_bounds} with $\delta = 1/3$.
\end{proof}

As an example, let $\cX = B^\Dim_\infty$ and $\Fun= \Fun_\Lip$ (cf.~Example~\ref{ex:lip}). Consider the case $r=1$. Let $\Lambda_\eps$ be a maximal $2\eps$-packing of $\cX$ in $\ell_2$. A simple volume counting argument shows that
\begin{align*}
|\Lambda_\eps| \ge v^{-1}_\Dim (1/\eps)^\Dim,
\end{align*}
where $v_\Dim = \vol(B^\Dim_2)$. Consider the subclass of $\Fun_\Lip$ consisting of all functions of the form $f_\theta(x) = \| x - \theta \|, \theta \in \Lambda_\eps$. Then for any two  distinct functions $f_\theta,f_{\theta'} $ we will have
\begin{align*}
\disc(f_\theta,f_{\theta'}) = \| \theta - \theta' \| \ge 2\eps,
\end{align*}
so $N(\Fun_\Lip,\disc,\eps) \ge v^{-1}_\Dim (1/\eps)^\Dim$. Theorem~\ref{thm:general_bound} then gives the following lower bound for any noisy oracle with $C^* < +\infty$:
\begin{align*}
K_\Pro(\eps) = \Omega\left(\Dim \log \frac{1}{\eps} \right).
\end{align*}
For noiseless first-order oracles, the same lower bound follows from a binary search argument, and can be achieved using the (computationally infeasible) {\em method of centers of gravity} \cite{NemYud83,Nes04}. In order to achieve this bound with a {\em noisy}  oracle, an algorithm must pose queries that reduce the uncertainty by an amount that is independent of $\eps$. This is possible with certain kinds of oracles \cite{BZ,castro2008minimax,NowakGBS}.

\subsection{First-order oracles with Gaussian noise}

If the oracle provides {\em noisy} first-order information, the above logarithmic lower bound can be tightened significantly. We now present lower bounds for two problem classes -- convex Lipschitz functions (cf.~Example~\ref{ex:lip_noisy}) and strongly convex functions (cf.~Example~\ref{ex:sc}) -- when the oracle supplies first-order information corrupted by additive white Gaussian noise. This is an oracle that, for a function $f$ and a query point $x$, responds with
\begin{align}\label{eq:1st_order_oracle}
Y = (f(x) + W, g(x) + Z),
\end{align}
where, as before, $g(x)$ is an arbitrary subgradient in $\partial f(x)$, and $W \sim \Normal(0,\sigma^2), Z \sim \Normal(0,\sigma^2 I_\Dim)$ are mutually independent. We will refer to this oracle as the {\em first-order Gaussian (FOG) oracle}. We will also consider the {\em subgradient-only Gaussian (SOG)} oracle $Y = g(x) + Z$. For simplicity, we will assume that the algorithm knows the structure of the deterministic selector mapping $(f,x) \mapsto g(x) \in \partial f(x)$, since this knowledge can only help. We will see that the $\eps$-complexities of these problem classes have {\em polynomial} dependence on $1/\eps$, but differ in their dependence on the problem dimension $\Dim$. Special cases of these results for linear functions for $\Dim=1$ and $r=1$ can be found, for example, in \cite{ShaNem05}. The exponent of $1/\eps$ in the bounds will, generally, depend on the smoothness of functions in $\Fun$. We remark that noise variance for each coordinate of the subgradient is a constant $\sigma^2$, implying that the expected squared $\ell_2$ norm of the noisy subgradient scales linearly with $\Dim$. We shall keep this in mind when considering achievability of the lower bounds by specific algorithms. It is also straightforward to treat the case when $Z \sim \Normal(0,(\sigma^2/n) I_\Dim)$, i.e., when $\sigma^2$ bounds the {\em total} (as opposed to per-coordinate) noise variance.

We begin by particularizing the IR bound \eqref{eq:IR_bound} to the oracles under consideration (the proof is given in Appendix~\ref{app:proofs}):

\begin{lemma}[IR bounds for Gaussian oracles]\label{lm:1st_order_IR_bound}
\begin{align}\label{eq:1st_order_IR_bound}
I(M; \wh{M}) \le	\begin{cases} \displaystyle\frac{T}{2\sigma^2} \max_{i,j} \sup_{x \in \cX} \Big\{ \left[f_i(x) - f_{j}(x) \right]^2 & \\
\qquad \qquad \qquad  + \| g_i(x) - g_{j}(x) \|^2 \Big\}  & \text{(FOG)} \\
 \displaystyle\frac{T}{2\sigma^2} \max_{i,j} \sup_{x \in \cX} \| g_i(x) - g_{j}(x) \|^2 & \text{(SOG)}
\end{cases}
\end{align}
\end{lemma}

We can now address the complexity of minimizing Lipschitz convex functions over a compact domain (cf.~Example~\ref{ex:lip_noisy}):
 
\begin{theorem}\label{thm:Lip_bounds} Consider the problem class $\Pro = (\cX,\Fun_\Lip,\Ora)$ with a Gaussian oracle, where $\cX \subset \Reals^n$ with $n \ge 16$. Define
\begin{align*}
s_\cX \deq \max \left\{ s \ge 0: s B^n_\infty \subseteq \cX \right\}.
\end{align*}
Then for any $r \ge 1$, $\eps \le \left(s_\cX \sqrt{n/8}\right)^r$, and $\delta \in (0,1/2)$, the following bounds hold:
\begin{align}\label{eq:Lip_bounds}
K^{(r)}_\Pro(\eps,\delta) \ge \begin{cases}
\displaystyle  \frac{((1-\delta)n -8) n s^2_\cX \log 2}{128(n s^2_\cX + 1)} \cdot \frac{\sigma^2}{\eps^{2/r}} & \text{(FOG)} \\
\displaystyle \frac{((1-\delta)n - 8) n s^2_\cX \log 2}{128} \cdot \frac{\sigma^2}{\eps^{2/r}} & \text{(SOG)}
\end{cases}
\end{align}
\end{theorem}

\begin{proof} By the Varshamov--Gilbert bound (see, e.g.,~Lemma~2.9 in \cite{Tsy09}), there exists an $\Dim/8$-packing of size $ N  > 2^{\Dim/8} \geq 4$ of the binary cube $\{-1,+1\}^\Dim$ in the Hamming distance. In other words, there exists a subset $\{ \xi_0,\ldots, \xi_{N-1} \}$ of the vertices of $B^n_\infty$ with $N > 2^{\Dim/8}$, such that
\begin{align}\label{eq:hamming}
\| \xi_i - \xi_{j} \|^2 &= 4 \sum^{\Dim}_{k=1} \ind{\xi_{i,k} \neq \xi_{j,k} } \ge \frac{n}{2}, \qquad \forall i \neq j.
\end{align}
Define the functions
\begin{align*}
f_i(x) \deq \frac{\eps^{1/r}}{s_\cX} \sqrt{\frac{8}{n}} \| x - s_\cX \xi_i \|, \quad i = 0,1,\ldots,N-1.
\end{align*}
Since $\eps \le \big(s_\cX \sqrt{n/8}\big)^{r}$ and $\{ s_\cX \xi_i \}^{N-1}_{i=0} \subset \cX$, these functions lie in $\Fun_\Lip$, and each $f_i$ is uniquely minimized at $x^*_i = s_\cX \xi_i$ with $f^*_i = 0$. Moreover, upon defining
\begin{align*}
d(f_i,f_{j}) \deq \eps^{1/r} \sqrt{\frac{8}{n}} \| \xi_i - \xi_{j} \|,
\end{align*}
it follows from \eqref{eq:hamming} and the triangle inequality that $d(f_i,f_{j}) \ge 2\eps^{1/r}$ for all $i \neq j$, and that this $d$ satisfies the condition \eqref{eq:pseudometric} on the set $\{f_i\}^{N-1}_{i=0}$. Hence, if there exist some $T \in \Naturals$ and an algorithm $\Alg \in \Algs_T(\Pro)$ that attains \eqref{eq:good_alg}, we can apply Lemma~\ref{lm:fano} to obtain
\begin{align}
I(M; \wh{M}_T) &\ge (1-\delta) \log N - \log 2 \nonumber \\
&> \frac{(1-\delta)n - 8}{8} \log 2. \label{eq:Lip_lower_bound}
\end{align}
Next we will bound $I(M; \wh{M}_T)$ from above. For any $x \in \cX$ and any pair $i,j$, we have
\begin{align*}
|f_i(x) - f_{j}(x)|^2 &\le \frac{8\eps^{2/r}}{n} \| \xi_i - \xi_{j} \|^2 \\
&= \frac{8\eps^{2/r}}{n} \sum^n_{k=1} (\xi_{i,k} - \xi_{j,k})^2 \\
&=  \frac{32 \eps^{2/r}}{n} \sum^n_{k=1} \ind{\xi_{i,k} \neq \xi_{j,k}} \\
&\le 32\eps^{2/r},
\end{align*}
and any subgradient of $f_i$ at $x$ has $\ell_2$ norm not exceeding $(\eps^{1/r}/s_\cX)\sqrt{8/n}$. Hence, applying Lemma~\ref{lm:1st_order_IR_bound}, we get
\begin{align}\label{eq:Lip_upper_bound}
I(M;\wh{M}_T) \le
\begin{cases} \displaystyle\frac{16 T \eps^{2/r}}{\sigma^2} \left( 1 + \frac{1}{n s^2_\cX} \right) & \text{(FOG)} \\
\displaystyle\frac{16 T \eps^{2/r}}{n \sigma^2 s^2_\cX} & \text{(SOG)}
\end{cases}
\end{align}
Combining \eqref{eq:Lip_lower_bound} and \eqref{eq:Lip_upper_bound} and rearranging, we get \eqref{eq:Lip_bounds}.
\end{proof}
Upper bounds on stochastic gradient descent -- an algorithm which only uses the subgradient information --  for $r=1$ are of the form $O\left(G^2 D^2_\cX/\eps^2 \right)$, where $G^2$ is an upper bound on the expected squared norm of the noisy gradient \cite{NJLS09}. As we show below, this is matched by our lower bounds. Indeed, $G^2 \propto \Dim\sigma^2$ for the additive Gaussian noise with variance $\sigma^2$. For the unit sphere we thus obtain $\Omega(\Dim\sigma^2/\eps^2)$; for the unit hypercube we obtain $\Omega\left(\Dim^2\sigma^2/\eps^2\right)$ for the SOG oracle:

\begin{corollary} Under the conditions of Theorem~\ref{thm:Lip_bounds}, we have
	\begin{align*}
		\cX = \rho B^n_\infty &\,\, \Longrightarrow \,\,
K^{(r)}_\Pro(\eps) = \begin{cases}
\displaystyle  \Omega \left( \frac{ n^2\rho^2}{1+n\rho^2} \cdot \frac{\sigma^2}{\eps^{2/r}}\right) & \text{(FOG)} \\
\displaystyle \Omega \left( n^2\rho^2 \cdot \frac{\sigma^2}{\eps^{2/r}}\right) & \text{(SOG)}
\end{cases}
\\
\cX = \rho B^n_2 &\,\, \Longrightarrow \,\,
K^{(r)}_\Pro(\eps) = \begin{cases}
\displaystyle  \Omega \left( \frac{ n\rho^2}{1+\rho^2}\cdot \frac{\sigma^2}{\eps^{2/r}}\right) & \text{(FOG)} \\
\displaystyle \Omega \left( n\rho^2  \cdot\frac{\sigma^2}{\eps^{2/r}}\right) & \text{(SOG)}
\end{cases}
\end{align*}
\end{corollary}

When the functions in $\Fun$ are {\em strongly convex} (cf.~Example~\ref{ex:sc}), rather than convex Lipschitz, the complexity of optimization will decrease:

\begin{theorem}\label{thm:SC_bounds} Consider the problem class $\Pro = (\cX,\Fun_\SC,\Ora)$ with a first-order or gradient-only Gaussian oracle, where $\cX \subset \Reals^n$ with $n \ge 16$. Then for any $r \ge 1$, $\eps \le \big(n s^2_\cX/16\big)^r$, $\delta \in (0,1/2)$, the following bounds hold:
\begin{align}\label{eq:SC_bound}
K^{(r)}_\Pro(\eps,\delta) \ge \begin{cases}
\displaystyle \frac{((1-\delta)n - 8) \log 2}{256 (D^2_\cX + 1)} \cdot \frac{\sigma^2}{\eps^{1/r}}, & \text{(FOG)} \\
\displaystyle \frac{((1-\delta)n - 8) \log 2}{256} \cdot \frac{\sigma^2}{\eps^{1/r}}, & \text{(SOG)}
\end{cases}
\end{align}
\end{theorem}

\begin{proof} Given $n$, construct the set $\{\xi_0,\ldots,\xi_{N-1}\} \subset \{-1,+1\}^n$ as in the proof of Theorem~\ref{thm:Lip_bounds} and define the functions
\begin{align*}
f_i(x) \deq \frac{1}{2} \left\| x -  \sqrt{\frac{16 \eps^{1/r}}{n}} \xi_i \right\|^2, \,\, i = 0,1,\ldots,N-1.
\end{align*}
Since $\eps \le \big(n s^2_\cX/16)^r$,
\begin{align*}
\left\{ x^*_i \deq  \sqrt{\frac{16 \eps^{1/r}}{n}} \xi_i \right\}^{N-1}_{i=0} \subset  \sqrt{\frac{16 \eps^{1/r}}{n}} B^n_\infty \subseteq s_\cX B^n_\infty \subseteq \cX.
\end{align*}
Thus, the functions $f_i$ lie in $\Fun_\SC$, and each $f_i$ is uniquely minimized by $x^*_i$ with $f^*_i = 0$. Moreover, upon defining
\begin{align*}
d(f_i,f_{j}) \deq \frac{4\eps^{1/r}}{n} \| \xi_i - \xi_{j} \|^2, 
\end{align*}
it follows from \eqref{eq:hamming} and the triangle inequality that $d(f_i,f_{j}) \ge 2\eps^{1/r}$ for all $i \neq j$, and that this $d$ satisfies \eqref{eq:pseudometric} on the set $\{f_i\}^{N-1}_{i=0}$. Hence, if there is some $T \in \Naturals$ and some $\Alg \in \Algs_T(\Pro)$ satisfying \eqref{eq:good_alg}, we can apply Lemma~\ref{lm:fano} to obtain \eqref{eq:Lip_lower_bound}.

Now we will derive an upper bound on $I(M; \wh{M}_T)$. For any $x \in \cX$ and any pair $i,j$, we have
\begin{align*}
	&\left|f_i(x) - f_j(x)\right| \\
	&= \frac{1}{2}\left| \| x - x^*_i \| + \| x - x^*_j \| \right| \cdot \left| \| x - x^*_i \| - \| x - x^*_j \| \right| \\
	&\le D_\cX \| x^*_i - x^*_j \|,
\end{align*}
where the last step uses the definition of $D_\cX$ and the triangle inequality. Thus,
\begin{align*}
|f_i(x) - f_{j}(x)|^2 \le D^2_\cX \| x^*_i - x^*_{j} \|^2 \le 64 D^2_\cX \eps^{1/r}.
\end{align*}
Also,
\begin{align*}
\| \nabla f_i(x) - \nabla f_{j}(x) \|^2 &= \| x^*_i - x^*_{j} \|^2 \le 64 \eps^{1/r}.
\end{align*}
Therefore, applying Lemma~\ref{lm:1st_order_IR_bound}, we get
\begin{align}\label{eq:SC_upper_bound}
I(M; \wh{M}_T) \le \begin{cases} 
\displaystyle \frac{32  (D^2_\cX + 1) T \eps^{1/r}}{\sigma^2} & \text{(FOG)} \\
\displaystyle \frac{32 T \eps^{1/r}}{\sigma^2} &\text{(SOG)}
\end{cases}
\end{align}
Combining \eqref{eq:Lip_lower_bound} and \eqref{eq:SC_upper_bound}, we get \eqref{eq:SC_bound}.
\end{proof}

\begin{corollary} Under the conditions of Theorem~\ref{thm:SC_bounds}, we have
	\begin{align*}
		\cX = \rho B^n_\infty &\,\, \Longrightarrow \,\,
K^{(r)}_\Pro(\eps) = \begin{cases}
\displaystyle  \Omega \left( \frac{n}{n\rho^2+1 } \cdot \frac{\sigma^2}{\eps^{1/r}}\right) & \text{(FOG)} \\
\displaystyle \Omega \left( n \cdot \frac{\sigma^2}{\eps^{1/r}} \right) & \text{(SOG)}
\end{cases} 
\\
\cX = \rho B^n_2 &\,\, \Longrightarrow \,\,
K^{(r)}_\Pro(\eps) = \begin{cases}
\displaystyle  \Omega \left( \frac{n}{\rho^2+1} \cdot \frac{\sigma^2}{\eps^{1/r}} \right) & \text{(FOG)} \\
\displaystyle \Omega \left( n \cdot \frac{\sigma^2}{\eps^{1/r}} \right) & \text{(SOG)}
\end{cases}
\end{align*}
\end{corollary}

\subsection{Noisy oracles satisfying a moment bound}

Our information-theoretic technique can be used to give a simpler derivation of the lower bounds obtained by Nemirovski and Yudin \cite[Ch.~5]{NemYud83} for Lipschitz convex functions and noisy first-order oracles satisfying a certain moment constraint.

Let $\cX = B^\Dim_\infty$ and $\Fun = \Fun_\Lip$, and consider the class of all noisy first-order oracles whose output $Y = (V^{0},V^{1}) \in  \Reals \times \Reals^\Dim$ satisfies the following two conditions:
\begin{itemize}
\item (C1) It is unbiased, i.e.,\ $\E[V^{0}|f,x] = f(x), \ \E[V^{1}|f,x] \in \partial f(x), \forall f \in \Fun, x \in \cX$.
\item (C2) There exist constants $\alpha > 1, L > 0$, such that
\begin{align*}
\E\big[|V^{0} - f(x)|^\alpha \big| f,x\big] \le L^\alpha, \quad \E \big[ \| V^{1} \|^\alpha \big| f,x \big] \le L^\alpha
\end{align*}
for all $f \in \Fun, x \in \cX$.
\end{itemize}
We will denote the class of all such oracles by $\Pi(\alpha,L)$.

\begin{theorem}\label{thm:rth_moment} There exists an oracle $\Ora \in \Pi(\alpha,L)$, such that the corresponding problem class $\Pro = (\cX, \Fun_\Lip,\Ora)$ satisfies
\begin{equation}\label{eq:rth_moment}
K_{\Pro}(\eps,\delta) \ge \frac{\log 2 - h_2(\delta)}{c \log 2} \eps^{-\alpha/(\alpha-1)}
\end{equation}
for all $\eps \in (0,\min\{L/2^{1/\alpha},1\}), \delta \in (0,1/2)$ with some $c = c(\alpha,L) > 0$.
\end{theorem}

\begin{proof} Define two functions $f_0(x) = -\xi^\tr x$ and $f_1(x) = \xi^\tr x$, where $\xi \in \Reals^\Dim$ has all coordinates equal to $\eps/\Dim$, and consider the following noisy oracle defined by Nemirovski and Yudin \cite[p.~198]{NemYud83}. Choose a constant $c > 0$ such that $c^{1-\alpha} < \min\{L^\alpha/2,1\}$, and let $p_{\eps,\alpha} \deq c \eps^{\alpha/(\alpha-1)}$.  On the set $\Fun \backslash \{f_0,f_1\}$, this oracle acts noiselessly, while on the set $\{f_0,f_1\}$ it acts as follows: given $f_i$, $i\in\{0,1\}$, and $x \in \cX$, it outputs
\begin{align*}
Y =  \begin{cases}
(0,0), & \text{with probability } 1-p_{\eps,\alpha} \\
p^{-1}_{\eps,\alpha}(f_i(x), \nabla f_i(x)), & \text{with probability } p_{\eps,\alpha}
\end{cases}
\end{align*}
It is an easy exercise to show that this oracle belongs to $\Pi(\alpha,L)$; moreover, on the set $\{f_0,f_1\}$ this oracle can be realized in the form \eqref{eq:SLO} with $\psi(f_i,x) = (f_i(x),\nabla f_i(x))$, $i \in \{0,1\}$.

Consider an algorithm $\Alg$ that achieves \eqref{eq:good_alg} with $r=1$. Let $U_t = \psi(f_M,X_t)$. Then $
I(U_t; Y_t |X^t, Y^{t-1}) \le I(U_t; Y_t | X_t)$ because $(X^{t},Y^{t-1}) \to U_t \to Y_t$ is a Markov chain. Now, given $X_t = x_t$, $U_t$ can take only two values, namely $(-\xi^\tr x_t, -\xi)$ or $(\xi^\tr x_t, \xi)$. Thus, $H(U_t|X_t) \le \log 2$. Moreover, since the mutual information $I(A;B|C)$ is convex in $P_{B|A,C}$, we have 
\begin{align*}
I(U_t; Y_t | X_t) \le  p_{\eps,\alpha} H(U_t|X_t) \le p_{\eps,\alpha}\log 2.
\end{align*}
Summing over the $T$ rounds and using Lemma~\ref{lm:info_bounds}, we get
\begin{align*}
I(M; \wh{M}_T) \le \sum^T_{t=1} I(U_t; Y_t |X_t) \le T c \eps^{\alpha/(\alpha-1)} \log 2.
\end{align*}
From Lemma~\ref{lm:fano}, we have $I(M; \wh{M}_T) \ge \log 2 - h_2(\delta)$. Combining these bounds and rearranging, we get (\ref{eq:rth_moment}).
 \end{proof}
 
\noindent The statement of Theorem~\ref{thm:rth_moment} should be interpreted in the following sense (cf.~also \cite{NemYud83}): given $\cX$ and $\Fun$ as above,
\begin{align*}
 \sup_{\Ora \in \Pi(\alpha,L)} K_{(\cX,\Fun_\Lip,\Ora)}(\eps) = \Omega(\eps^{-\alpha/(\alpha-1)}).
 \end{align*}
Thus, we have a lower bound which is robust relative to $\Pi(\alpha,L)$. However, this bound is sharp only for $\alpha \in (1,2]$ \cite{NemYud83,Aga09,Aga10}: for $\alpha \ge 2$, the correct bound is $\Omega(1/\eps^2)$. This can be easily seen from the results of the preceding section on the Gaussian first-order oracle.

\subsection{Statistical estimation and sequential experimental design}
\label{ssec:stats}

Finally, let us see how the problems of parametric statistical estimation (Example~\ref{ex:statistical}) and experimental design (Example~\ref{ex:design}) can be viewed through the lens of optimization.

Let us consider statistical estimation first. We will use the notation of Example~\ref{ex:statistical}. Typically, one considers the setting in which the statistician gets $T$ i.i.d.\ samples $Y_1,\ldots,Y_T$ drawn from some $P_\theta$, where $\theta$ is unknown. The quantity of interest is the {\em minimax risk}
\begin{align*}
	R^*_T \deq \inf_{\wh{\theta}_T} \sup_{\theta \in \cX} \E_\theta \left(f_\theta(\wh{\theta}_T) - f_\theta(\theta)\right)^r,
\end{align*}
where the infimum is over all measurable estimators $\wh{\theta}_T : \cY^T \to \cX$, and $f_\theta$ is an element of an appropriate class $\Fun$ of loss functions, such as \eqref{eq:loss} or \eqref{eq:contrast}.

Now, the output of any such estimator can be viewed as the final result of some algorithm $\Alg \in \Algs_T(\Pro)$. Thus, we simply follow our general recipe and isolate a finite subset $\Fun' = \{ f_{\theta_0},\ldots,f_{\theta_{N-1}}\}$ such that, with a suitably defined ``distance'' $\disc(\cdot,\cdot)$ that satisfies \eqref{eq:pseudometric}, we have
\begin{align*}
	\disc(f_{\theta_i},f_{\theta_j}) \ge 2\eps^{1/r}, \qquad \forall i \neq j.
\end{align*}
Suppose that we can arrange things in such a way that the cardinality of such an $\Fun'$ is independent of $\eps$, but may still depend on $\Dim$ and $r$: $N = N(r,\Dim)$ (this is possible in many cases, cf.~the proofs of Theorems~\ref{thm:Lip_bounds} and \ref{thm:SC_bounds}). Then we simply apply the IR bound to get
\begin{align}\label{eq:statistical_complexity}
	K^{(r)}_\Pro(\eps) = \Omega \left( \frac{\log N(r,\Dim)}{\max_{i,j} D(P_{\theta_i}\| P_{\theta_j})}\right).
\end{align}
Assuming, as is often the case, that
\begin{align*}
	\max_{i,j} D(P_{\theta_i}\|P_{\theta_j}) = C(r,\Dim)\eps^{\gamma/r}
\end{align*}
for some $C(r,\Dim) > 0, \gamma > 0$, we can invert \eqref{eq:statistical_complexity} and obtain the minimax lower bound
\begin{align*}
R^*_T = \Omega\left( \left( \frac{\log N(r,\Dim)}{C(r,\Dim)}\right)^{r/\gamma} T^{-r/\gamma}
	\right).
\end{align*}

Similar considerations apply to sequential experimental design as well (Example~\ref{ex:design}), except there we have a {\em design strategy} $\{ \wh{X}_t : \cY^{t-1} \to \cX\}^T_{t=1}$ and the estimator $\wh{\theta}_T : \cY^T \to \cX$, where at time $t = 1,\ldots,T$ we choose the design point $X_t = \wh{X}_t(Y^{t-1})$ and obtain a sample $Y_t \sim P_{\theta,X_t}$, and then at time $T+1$ we process all the samples to get the estimate $X_{T+1} = \wh{\theta}_T(Y^T)$. The minimax risk is then
\begin{align*}
	R^*_T = \inf_{\wh{X}^T,\wh{\theta}_T} \sup_{\theta \in \cX} \E\left( f_\theta(\wh{\theta}_T) - f_\theta(\theta)\right)^r.
\end{align*}
The connection to optimization is even more apparent than in the estimation setting, and the IR bound technique yields
\begin{align*}
	K^{(r)}_\Pro(\eps) = \Omega \left( \frac{\log N(r,\Dim)}{\max_{i,j}\sup_{x \in \cX}D(P_{\theta_i,x}\| P_{\theta_j,x})}\right).
\end{align*}
Just as before, this bound can be inverted to get a lower bound on the minimax risk. Note, however, that what we have is a lower bound on the number of the design points needed to guarantee that the minimax risk is below $\eps$.

\section{Lower bounds for anytime algorithms}
\label{sec:anytime}

Conceptually, our use of the IR bounds is akin to the methods used in statistics to obtain minimax lower bounds through local entropy estimates and a device like the Assouad lemma (cf.~\cite{Yu97} and references therein). In both cases, in order to get the right rates it is essential to arrange things so that the size of the ``packing" set $\Fun'$ is independent of $\eps$. However, one drawback of the IR bounds is that they do not take into account the {\em dynamics} of the algorithm, pertaining to the manner in which its expected error evolves with time. Instead, we must use uniform, worst-case bounds on the uncertainty remaining after each successive oracle call. However, it could be argued on practical grounds that the only optimization algorithms that are of any value are the ones whose performance gradually and monotonically improves with time, as more and more queries are issued --- that is, anytime algorithms. In this section, we show that the LF bounds can be used to track the evolution of the mutual information over time. As a consequence, we will be able to derive upper bounds on the anytime exponent for certain problem classes.

We will show that the amount of information extracted by an anytime algorithm at each time step obeys a {\em law of diminishing returns:} as the queries $X_t$ approach the minimizer, the rate at which the algorithm can reduce its uncertainty about the objective function slows down. Moreover, assuming that the worst-case expected error of such an algorithm decays polynomially with time, we will obtain lower bounds on the rate of this decay. We will also show that, in some cases, insisting on the anytime property may mean that the algorithm will take {\em longer} to get to the point after which its expected error drops below some desired level. This seemingly strange conclusion reflects the fact that, without placing any restrictions on the algorithm's trajectory, we are allowing ``bizarre'' (and not very practical) strategies that wander around the problem domain for a while, gathering information without much regard to how close they are to a minimizer, and then --- boom! --- produce an excellent solution. With such algorithms, it is certainly no surprise that they may hit upon a good solution more quickly than an ``honest'' anytime algorithm that must proceed incrementally and inexorably towards a minimizer.

In contrast to the local technique based on the IR bounds, our use of the LF bounds in this section can be thought of as a {\em global technique} \cite{YanBar99,YangAISTATS,Guntuboyina}. The main idea is as follows. Suppose we have an anytime algorithm whose worst-case expected errors decay at some rate $\eps_t \to 0$. Then, for each $T$, we consider an $\eps_T$-packing of the problem domain (with respect to a suitable metric, typically just the usual Euclidean norm $\| \cdot \|$), which will induce a packing of the function class $\Fun$. This packing will be of size $\Omega(\eps^{-\Dim}_T)$. Since the algorithm does well on every single function in $\Fun$, it must necessarily do well on every function in this large packing set. Thus, if the objective function is drawn uniformly at random from this set, then combining the lower bound of Lemma~\ref{lm:fano} with the LF upper bound will result in a relation of the form
\begin{align*}
\Dim	\log \left(\frac{1}{\eps_T}\right) \preceq \sum^T_{t=1} \eps^\gamma_t,
\end{align*}
where $\gamma > 0$ depends on the smoothness of the functions in $\Fun$. This relation must hold for all but finitely many values of $T$. The optimal rate is then derived by balancing the entropy $\Dim \log \eps^{-1}_T$ and the sum of diminishing mutual information terms.

\subsection{Strongly convex functions}
\label{ssec:strongly_convex}

We first consider the case of strongly convex functions with Lipschitz-continuous gradients (an often made assumption \cite{Vilk,NJLS09}). Given $\cX$, let $\Fun_{\kappa,L}$ denote the set of all functions $f : \cX \to \Reals$ that satisfy the following conditions:
\begin{itemize}
\item Each $f \in \Fun_{\kappa,L}$ is $\kappa$-strongly convex (cf.~Example~\ref{ex:sc}).
\item For each $f$, the mapping $x \mapsto \nabla f(x)$ is $L$-Lipschitz:
\begin{align*}
\| \nabla f(x) - \nabla f(y) \| \le L \| x - y \|, \,\, \forall x,y \in \cX.
\end{align*}
\end{itemize}
Consider the noisy first-order oracle $Y = (f(x) + W, \nabla f(x) + Z)$ as in \eqref{eq:1st_order_oracle}. Let $\{f_0,\ldots,f_{N-1}\} \subset \Fun_{\kappa,L}$ be a finite set of functions such that $f^*_0 = \ldots = f^*_{N-1} = c^*$ and $\nabla f_i(x^*_i) = 0$, where $x^*_i$ is the (unique) minimizer of $f_i$ on $\cX$.  Let $M$ denote the uniform random variable on $\{0,\ldots,N-1\}$. Then we have the following (the proof is in Appendix~\ref{app:proofs}):

\begin{lemma}\label{lm:LF_SC} At every time $t=1,2,\ldots$, any algorithm $\Alg \in \Algs_\infty(\Pro)$ satisfies
\begin{align}\label{eq:LF_SC}
 I(M; Y_t | X^t,Y^{t-1})  \le \frac{(L/\kappa)^2\left(D_\cX^2 + 1\right)}{\sigma^2}   \max_i \E \err_\Alg(t,f_i).
\end{align}
Moreover, if $\Alg$ is $1$-anytime (cf.~Definition~\ref{def:anytime}), then
\begin{align}
I(M; Y_t | X^t,Y^{t-1}) &\le \frac{(L/\kappa)^2}{\sigma^2} (D^2_\cX + 1) \cdot \Err_\Alg(t,\Fun_{\kappa,L}) \nonumber \\
&\xrightarrow{t \to \infty} 0. \label{eq:dim_returns}
\end{align} 
\end{lemma}

\noindent Thus, the decay of the expected error in minimizing a strongly convex function is accompanied by the decay of the average information gain, and, moreover,  the two quantities decay at the same rate. In other words, anytime algorithms for strongly convex programming obey a {\em law of diminishing returns}. Evidently, this phenomenon is due to the fact that, as the algorithm zeroes in on the minimizer, the signal-to-noise ratio keeps dropping because the mean-square error and the mean-square norm of the gradient both decrease as $O(\eps_t)$. Using Lemma~\ref{lm:LF_SC} in conjunction with the information bounds of Section~\ref{sec:lower_bounds}, we establish the following upper bound on the anytime exponent of strongly convex programming problems:

\begin{theorem}\label{thm:anytime_SC} Consider the problem class $\Pro = (\cX,\Fun,\Ora)$ with $\cX = B^\Dim_\infty$, $\Fun = \Fun_{\kappa,L}$ with $\kappa = 1$ and $L \ge 1$, and the Gaussian first-order oracle \eqref{eq:1st_order_oracle}. Then
$\gamma_\Pro \le 1$. In other words, on this problem class, $O(t^{-1})$ is the optimal error decay rate for all $1$-anytime algorithms whose errors decay polynomially with $t$.
\end{theorem}

\begin{proof} Consider any algorithm $\Alg \in \Algs_\infty(\Pro)$ whose worst-case errors $\eps_t \deq \Err_\Alg(t,\Fun)$  satisfy
\begin{align*}
\limsup_{t \to \infty} t^\gamma \eps_t < \infty
\end{align*}
for some $\gamma \ge 0$. In other words, $\eps_t = O(t^{-\gamma})$. By Markov's inequality, we have, for every $T,$
\begin{align*}
\sup_{f \in \Fun} \Pr \Big( \err_\Alg(T,f) \ge 3\eps_T \Big) \le \frac{\sup_{f \in \Fun} \E \err_\Alg(T,f)}{3\eps_T} \le \frac{1}{3}.
\end{align*}
Let us fix some $T$, let $\Lambda_T = \{\theta_0,\ldots,\theta_{N-1}\}$ denote a maximal $2\sqrt{3 \eps_T}$-packing set in $\cX$ (w.r.t.\ $\| \cdot \|$), and define
\begin{align*}
f_i(x) \deq \frac{1}{2} \|x-\theta_i \|^2, \qquad i=0,\ldots,N-1.
\end{align*}
By volume counting, $N \ge v^{-1}_\Dim (1/3\eps_T)^{\Dim/2}$. We also have $\disc(f_i,f_{j}) = \frac{1}{2} \| \theta_i - \theta_{j} \|^2 \ge 6\eps_T$. By Lemma~\ref{lm:fano},
\begin{align}\label{eq:step1d}
I(M; \wh{M}_T) \ge  \frac{\Dim}{3} \log \left(\frac{1}{\eps_T}\right) + c_\Dim,
\end{align}
where $c_\Dim = \frac{1}{3}\log \left(\frac{1}{3^\Dim 8 v^2_\Dim}\right)$. On the other hand, applying Lemma~\ref{lm:LF_SC}, we obtain
\begin{align}\label{eq:step2d}
I(M; \wh{M}_T) \le \frac{\Dim+1}{\sigma^2} \sum^T_{t=1} \eps_t.
\end{align}
Combining (\ref{eq:step1d}) and (\ref{eq:step2d}), we see that the sequence $\{\eps_t\}$ must satisfy the following inequalities:
\begin{align*}%\label{eq:step3d}
\frac{\Dim \sigma^2}{3(\Dim+1)}\log \left(\frac{1}{\eps_T}\right) + c'_\Dim \le  \sum^T_{t=1} \eps_t, \quad \forall T
\end{align*}
where $c'_\Dim = \sigma^2 c_\Dim/(\Dim+1)$. From Lemma~\ref{lm:fun_rec} we therefore conclude that there exists an infinite subsequence of times $1 \le t_1 < t_2 < \ldots$, such that $\eps_{t_j} \ge c t^{-1}_j$ for some $c > 0$. Since $\eps_t = O(t^{-\gamma})$ by hypothesis, we must have $\gamma \le 1$.
\end{proof}

\noindent The bound $\Omega(t^{-1})$ is tight and can be achieved by stochastic gradient descent \cite{NJLS09}. Note that the methods of Section~\ref{sec:lower_bounds} can be used to explicitly identify the dependence of the lower bound on the problem dimension $\Dim$.

\subsection{Comparison of IR and LF bounds}

A natural question is whether LF bounds for anytime algorithms provide tighter lower bounds when compared to IR bounds without the anytime assumption. This is indeed the case, as we demonstrate through an example. Interestingly, the difference is not present for linear and quadratic functions, but appears for higher degree polynomials. Consider a simple set-up with $\cX = [-1,1]$,
\begin{align*}
\Fun = \Fun_m = \left\{ f_\theta(x) = (x-\theta)^m : \theta \in \cX \right\}
\end{align*}
for some even $m \in \Naturals$, and the noisy first-order oracle
\begin{align*}
Y = (f(x) + W, f'(x) + Z)
\end{align*}
where $W,Z$ are mutually independent $\Normal(0,\sigma^2)$ random variables.

\begin{theorem}\label{thm:poly_arbitrary} On this problem class, for any $\eps \le 1$ and any $\delta \in (0,1/2)$ we have
\begin{align}\label{eq:rpower_arbitrary}
K^{(2)}_\Pro(\eps) = \Omega \left( \frac{\sigma^2}{ 4^m m^4 \eps^{1/m}} \right).
\end{align}
\end{theorem}

\begin{proof} Let us define
\begin{align}\label{eq:disc_rpower}
d(f_\theta,f_{\theta'}) \deq 2^{1-m} (\theta - \theta')^m, \qquad \forall \theta, \theta' \in \cX.
\end{align}
By convexity of $x \mapsto x^m$ (recall that $m$ is even), we have for any $x,\theta,\theta' \in \cX$
\begin{align*}
\left( \theta - \theta' \right)^m &= 2^m \left( \frac{\theta - x}{2} + \frac{ x - \theta'}{2}\right)^m \\
&\le 2^{m-1} \left[ (x-\theta)^m + (x-\theta')^m\right].
\end{align*}
Hence, the $d(\cdot,\cdot)$ defined in \eqref{eq:disc_rpower} satisfies \eqref{eq:pseudometric}. In particular, if we fix the functions
\begin{align*}
f_0(x) = (x-\eps^{1/2m})^m \,\, \text{and} \,\, f_1(x) = (x + \eps^{1/2m})^m
\end{align*}
then $d(f_0,f_1) = 2\eps^{1/2}$. Let $M$ have uniform distribution on $\{0,1\}$. Consider now any algorithm $\Alg \in \Algs_T(\Pro)$ that attains $\err^2_\Alg(T,f_\theta) < \eps$ with probability at least $1-\delta$ for every $\theta \in \cX$. Applying Lemma~\ref{lm:fano}, we obtain
\begin{align}\label{eq:rpower_lower}
I(M; \wh{M}_T) \ge \log 2 - h_2(\delta).
\end{align}
On the other hand, from Lemma~\ref{lm:1st_order_IR_bound} we have
\begin{align*}
 I(M; \wh{M}_T) \le \frac{1}{2\sigma^2} \sup_{x \in \cX} \Big\{ [f_0(x) - f_1(x)]^2  + [f'_0(x) - f'_1(x)]^2 \Big\}.
\end{align*}
From convexity, 
\begin{align*}
[f_0(x) - f_1(x)]^2 &= [ (x - \eps^{1/2m})^m - (x+\eps^{1/2m})^m ]^2 \\
&\le m^2 4^m \eps^{1/m}.
\end{align*}
Likewise, applying the mean-value theorem to the function $x \mapsto x^{m-1}$, we get
\begin{align*}
	& [f'_0(x) - f'_1(x)]^2 \\
	&= m^2 \left[ (x-\eps^{1/2m})^{m-1} - (x+\eps^{1/2m})^{m-1}\right]^2\\
	& \le [m(m-1)]^2 2^{2(m-1)} \eps^{1/m}.
\end{align*}
Thus, we obtain the upper bound
\begin{align}\label{eq:rpower_upper}
I(M; \wh{M}_T) = O \left( \frac{T 4^m m^4 \eps^{1/m}}{\sigma^2}\right).
\end{align}
Combining \eqref{eq:rpower_lower} and \eqref{eq:rpower_upper}, we obtain \eqref{eq:rpower_arbitrary}.
\end{proof}

\begin{theorem}\label{thm:poly_anytime} Consider the same problem class. Then $\gamma^{(2)}_\Pro \le \frac{m}{m-1}$.
\end{theorem}

\begin{proof} As before, consider any algorithm $\Alg \in \Algs_\infty(\Pro)$ whose worst-case errors $\eps_t \deq \Err^{(2)}_\Alg(t,\Fun)$  satisfy
\begin{align*}
\limsup_{t \to \infty} t^\gamma \eps_t < \infty
\end{align*}
for some $\gamma \ge 0$. Applying the same argument based on Markov's inequality as in the proof of Theorem~\ref{thm:anytime_SC}, we see that, for any $T$, $\Alg$ attains $\err^2_\Alg(T,f_\theta) < 3\eps_T$ with probability at least $2/3$ on every $f_\theta \in \Fun$. Given $T$, let $\Lambda_T = \{\theta_0,\ldots,\theta_{N-1}\}$ denote the largest finite subset of $\cX = [-1,1]$, such that
\begin{align*}
| \theta_i - \theta_{j} | \ge 2(3\eps_T)^{1/2m}, \qquad \forall i \neq j.
\end{align*} 
A simple counting argument shows that $N \ge \left(\frac{1}{3\eps_T}\right)^{1/2m}$.
Moreover, the functions $f_i(x) \deq (x-\theta_i)^m,  i=0,\ldots,N-1$, satisfy $d(f_i,f_j) \ge 2\sqrt{3\eps_T}$ for $i\neq j$, where $d(\cdot,\cdot)$ is defined in \eqref{eq:disc_rpower}.

By Lemma~\ref{lm:fano},
\begin{align*}
I(M; \wh{M}_T) \ge  \frac{1}{3m} \log \left(\frac{1}{\eps_T}\right) + c_m,
\end{align*}
where $c_m = \frac{1}{3}\log \left( \frac{1}{3^{1/m} 8}\right)$. We will now combine this lower bound with an appropriate LF bound. Let $Q^*$ denote the bivariate normal distribution $\Normal(0, \sigma^2 I_2)$. Then, for every $i = 0,\ldots,N-1$ we have
\begin{align*}
Y^{(i)} = (f_i(\theta_i) + W, f'_i(\theta_i) + Z) = (W,Z) \sim Q^*.
\end{align*}
Hence, applying the LF bound with $\ProbQ^*_{Y_t} = Q^*$  as in the proof of Lemma~\ref{lm:LF_SC}, we obtain
\begin{align}
& I(M; \wh{M}_T) \nonumber\\
&\le \frac{1}{2\sigma^2} \sum^T_{t=1} \E \left\{ f^2_M(X_t) + [f'_M(X_t)]^2\right\}  \nonumber\\
&= \frac{1}{2\sigma^2} \sum^T_{t=1} \E \left\{ (X_t - \theta_M)^{2m} + m^2 (X_t - \theta_M)^{2(m-1)} \right\} \nonumber\\
&= \frac{1}{2\sigma^2} \sum^T_{t=1} \E \Big\{ [f_M(X_t) - f^*_M]^2 \nonumber\\
& \qquad \qquad + m^2 \left[ f_M(X_t) - f^*_M \right]^{\frac{2(m-1)}{m}} \Big\} \nonumber\\
&\le \frac{1}{2\sigma^2} \sum^T_{t=1} \Big\{\E [f_M(X_t) - f^*_M]^2 \nonumber\\
& \qquad \qquad + m^2 \left( \E[f_M(X_t)-f^*_M]^2\right)^{\frac{m-1}{m}} \Big\} \label{eq:LF_IR_comp_1}\\
&\le \frac{m^2 + 1}{2\sigma^2} \sum^T_{t=1} \left( \E[f_M(X_t)-f^*_M]^2\right)^{\frac{m-1}{m}} \nonumber\\
&= \frac{m^2+1}{2\sigma^2}\sum^T_{t=1} \left( \E \err^2_\Alg(t,f_M)\right)^{\frac{m-1}{m}} \nonumber\\
&\le \frac{m^2+1}{2\sigma^2}\sum^T_{t=1} \eps^{\frac{m-1}{m}}_t, \nonumber
\end{align}
where in \eqref{eq:LF_IR_comp_1} we have used the concavity of the function $u \mapsto u^{(m-1)/m}$. Therefore, we conclude that the sequence $\{\eps_t\}$ must satisfy
\begin{align*}
	\frac{2\sigma^2}{3(m^3 + m^2)}\log \left(\frac{1}{\eps_T}\right) + \frac{2\sigma^2 c_m}{m^2 + 1} \le \sum^T_{t=1}\eps^{\frac{m-1}{m}}_t
\end{align*}
for all sufficiently large $T$. Applying Lemma~\ref{lm:fun_rec}, we conclude that there exists an infinite subsequence of times $t_1 < t_2 < \ldots$, such that $\eps_{t_j} \ge ct^{-m/(m-1)}_j$ for some constant $c > 0$. Since $\eps_t = O(t^{-\gamma})$ by hypothesis, we must have $\gamma \le \frac{m}{m-1}$.
\end{proof}

For $m=2$, the two results indicate the same order of complexity, $T \succeq \eps^{-1/2}$; however, for $m=4$ and larger, the bounds differ, giving $T \succeq \eps^{-1/m}$ for arbitrary algorithms and $T \succeq \eps^{(1-m)/m}$ for anytime algorithms, which is larger. We conclude that, in general, the LF bounding technique leads to tighter bounds for optimization algorithms which actually converge monotonically to the optimal solution.

\subsection{Active learning}
\label{ssec:active}

Our technique for analyzing anytime optimization algorithms  can also be used to give a particularly simple derivation of the minimax lower bound for active learning of a threshold function on the unit interval \cite{castro2008minimax}. In general, active learning is more difficult than (convex) optimization. However, for the case below, we can apply the tools developed in this paper. The reason for including this example is twofold: first, to show that problems beyond convex optimization can be attacked with our information-theoretic method, and second to exhibit a problem with a noise model more complicated than those encountered so far in the paper. 

The active learning problem is stated as follows. We have a pair $(X,Z)$ of jointly distributed random variables $X \in \cX = [0,1]$ and $Z \in \{0,1\}$, where the marginal distribution $P_X$ is uniform on $[0,1]$, while the conditional distribution $P_{Z|X}$ is unknown. We do, however, have some prior knowledge about $P_{Z|X}$. Define $\eta(x) \deq \E[Z|X=x]$. Then we assume the following:
\begin{itemize}
\item There exists some $\theta \in [0,1]$, such that $\eta(x) < 1/2$ for $x < \theta$ and $\eta(x) \ge 1/2$ otherwise. In other words, the {\em Bayes classifier} $G^*(x) \deq \ind{\eta(x) \ge 1/2}$ for this problem \cite{DGL_yellowbook} is of the form $G^*(x) = G_{\theta}(x) = \ind{x \ge \theta}$.
\item For some $0 < c < C < 1/2$ and $\kappa \in [1,\infty)$, we have
\begin{align}\label{eq:tsybakov}
c|x-\theta|^{\kappa-1} \le |\eta(x) - 1/2| \le C|x-\theta|^{\kappa-1},
\end{align}
where the first inequality (known as the {\em Tsybakov noise condition} \cite{Tsybakov_aggregation}) holds for all $x$ in a sufficiently small neighborhood of $\theta$.
\end{itemize}

Let $\Pi(\kappa,c,C)$ denote the class of all conditional probability distributions $P_{Z|X}$ satisfying these two conditions. We wish to determine the unknown {\em threshold} $\theta$ using an {\em active strategy}: at time $t$, we request a label $Z_t \in \{0,1\}$ at a point $X_t \in \cX$, chosen as a function of the history $(X^{t-1},Z^{t-1})$. Given our query $X_t$, the label $Z_t$ is generated at random according to $P_{Z|X}(\cdot|X_t)$. At time $t$, the candidate classifier is $G_{X_t}(x) = \ind{x \ge X_t}$. The performance of the strategy after $t$ time steps is measured by the {\em excess risk} relative to  $G^*$:
\begin{align}\label{eq:excess_risk}
R(G_{X_t}) - R(G^*) = \int_{[X_t,1] \triangle [\theta,1]} |2\eta(x)-1| dx,
\end{align}
where $\triangle$ denotes symmetric difference between sets. (The risk of a classifier $G : \cX \to \{0,1\}$ is defined as $R(G) \deq \Pr(G(X) \neq Z)$, and the Bayes risk is $R(G^*) \deq \inf_G R(G)$ \cite{DGL_yellowbook}.)

Castro and Nowak \cite{castro2008minimax} have shown that any active strategy will have excess risks of $\Omega(t^{-\kappa/(2\kappa-2)})$, and gave an explicit scheme that achieves the rate $O(t^{-\kappa/(2\kappa-2)})$. Their proof of the lower bound relies on an intricate construction of two distributions $P^{(1)}_{Z|X},P^{(2)}_{Z|X} \in \Pi(\kappa,c,C)$ that are close in a statistical sense, but far apart in the sense of their Bayes risks. We now show that the same lower bound can be derived using our machinery without any careful function tuning. To that end, we will cast this problem in the optimization setting, as alluded to in Example~\ref{ex:active}. Let $\cX$ and $\Fun$ be as described there, and associate to each $P_{Z|X} \in \Pi(\kappa,c,C)$ a noisy oracle with $\cY = \{-1,+1\}$ and $P(Y=1|f,x) = P(Y=1|\theta,x) = \eta(x)$. With this correspondence in place, we can now prove the following:

\begin{theorem} Let $\kappa \in (1,2]$. Suppose that there exists an active learning strategy satisfying
$$
\sup_{P_{Z|X} \in \Pi(\kappa,c,C)} \E[R(G_{X_t}) - R(G^*)] = O(t^{-\gamma})
$$
for some $\gamma > 0$. Then $\gamma \le \kappa/(2\kappa -2)$.  Thus, $O(t^{-\kappa/(2\kappa-2)})$ is the optimal decay rate for all active learning strategies whose excess risks decay as ${\rm Poly}(t^{-1})$. If $\kappa = 1$, then the excess risk is $\Omega(2^{-6C^2t})$.\footnote{The exponent in this lower bound is not tight, since there exists a specific strategy that achieves the excess risk of $O(2^{-c^2 t \log e/2 })$ when $\kappa =1$ \cite{castro2008minimax}.}
\end{theorem}

\begin{proof} For each $\theta \in [0,1]$, find some $P^{\theta}_{Z|X} \in \Pi(\kappa,c,C)$, such that the inequalities in (\ref{eq:tsybakov}) hold for {\em all} values of $x \in \cX$. Given a candidate classifier $G_{X_t}$, consider the excess risk $R(G_{X_t}) - R(G_{\theta})$. Assume for now that $\theta > X_t$. Then from (\ref{eq:excess_risk}) and (\ref{eq:tsybakov}) we get
\begin{align*}
R(G_{X_t}) - R(G_\theta) &\ge 2c \int^{\theta}_{X_t} (\theta-x)^{\kappa - 1} dx = \frac{2c}{\kappa} (\theta - X_t)^\kappa.
\end{align*}
The case $X_t < \theta$ is similar. Thus, the expected excess risk of any strategy at time $t$ can be bounded as
\begin{align}\label{eq:risk_bound}
\E[R(G_{X_t}) - R(G_\theta)] \ge (2c/\kappa) \E |X_t - \theta|^\kappa.
\end{align}
Now suppose we have a learning strategy whose worst-case excess risks decay at a prescribed rate $\{r_t\}$:
\begin{align*}
\sup_{P_{Z|X} \in \Pi(\kappa,c,C)} \E[R(G_{X_t}) - R(G^*)] = r_t, \,\, t=1,2,\ldots
\end{align*}
Then from this and (\ref{eq:risk_bound}) we have that, for every $P^\theta_{Z|X}$, this strategy satisfies
\begin{align}\label{eq:dist_bound}
 \E |X_t - \theta|^\kappa \le \kappa r_t/2c, \qquad t=1,2,\ldots
\end{align}
Let $\eps_t \deq (3\kappa r_t/2c)^{1/\kappa}$. Then using (\ref{eq:dist_bound}) and Markov's inequality, we see that for this strategy we must have
\begin{align*}
\sup_{\theta \in [0,1]} \Pr \big( |X_t - \theta| \ge \eps_t \big| \theta \big) \le 1/3, \,\, \forall t=1,2,\ldots.
\end{align*}
In other words, this active learning strategy gives rise to an {\em optimization algorithm} $\Alg$ for the problem class $\Pro=(\cX,\Fun,\Ora)$, where $\Ora$ is specified by $P(Y=1|\theta,x) = \E_{\theta}[Z|X=x]$, and there exists some $T_0 \ge 1$ such that $\Pr(\err_\Alg(t,f) \ge \eps_t) \le 1/3$ for all $t \ge T_0$.

Now for each $T \ge T_0$ let $\Lambda_T = \{\theta_0,\ldots,\theta_{N-1}\} $ be a maximal $2\eps_T$-packing of $[0,1]$. Simple counting shows that $N \ge 1/2\eps_T$. Consider the set $\Fun' = \{f_m = f_{\theta_m} : \theta \in \Lambda_T \} \subset \Fun$, and denote $\eta_m(x) \deq \E_{\theta_m}[Z|X=x]$. Then, in our usual notation, we have from Lemma~\ref{lm:fano} that
\begin{align}\label{eq:info_lower}
I(M; \wh{M}_T) \ge \frac{2}{3}\log \left(\frac{1}{\eps_T}\right) - \frac{5}{3}\log 2.
\end{align}
Next we apply Lemma~\ref{lm:info_bounds}. To that end, let us inspect the terms $I(M; Y_t | X^t, Y^{t-1})$:
\begin{align*}
& I(M; Y_t | X^t, Y^{t-1}) \nonumber\\
&= I(M,X_t; Y_t |X^{t-1},Y^{t-1}) - I(X_t; Y_t | X^{t-1},Y^{t-1}) \\
&\le I(M,X_t; Y_t |X^{t-1},Y^{t-1}) \le I(M,X_t; Y_t),
\end{align*}
where the first step uses the chain rule, the second is because mutual information is nonnegative, and the third is because $(X^{t-1},Y^{t-1}) \to (M,X_t) \to Y_t$ is a Markov chain. Now we use the LF bound with $Q^*$ the uniform distribution on $\{-1,+1\}$. Then
\begin{align}
I(M,X_t; Y_t) &\le D(\Prob_{Y_t|M,X_t} \| \ProbQ^*_{Y_t} | \Prob_{M,X_t}) \nonumber \\
&\le 4\E_{M,X_t} \left\{(\Prob(Y_t=1|M,X_t)-1/2)^2\right\}  \nonumber \\
&= 4\E_{M,X_t} \left\{ |\eta_M(X_t) - 1/2|^2 \right\} \nonumber \\
&\le 4C^2 \E_{M,X_t} |X_t - \theta_M|^{2(\kappa-1)},\label{eq:info_bound_2}
\end{align}
where in the second step we used the fact that
\begin{align*}
d(p\|1/2) \deq p\log 2p + (1-p)\log [2(1-p)] \le 4(p-1/2)^2
\end{align*}
for all $p \in [0,1]$, and in the last step we used (\ref{eq:tsybakov}). Suppose first that $\kappa > 1$. Because $\kappa \le 2$, the function $x \mapsto x^{(2\kappa-2)/\kappa}$ is concave, and we can write
\begin{align*}
\E |X_t - \theta_M|^{2(\kappa-1)} \le \left(\E | X_t - \theta_M|^\kappa \right)^{2(\kappa-1)/\kappa}.
\end{align*}
Using this in conjunction with (\ref{eq:dist_bound}) and Lemma~\ref{lm:info_bounds}, we can bound the mutual information $I(M; \wh{M}_T)$ as
\begin{align}
I(M; \wh{M}_T) \le 4C^2 \sum^T_{t=1}\left(\frac{\kappa r_t}{2c}\right)^{\frac{2(\kappa-1)}{\kappa}} = \frac{4C^2}{3^{\frac{2\kappa-2}{\kappa}}} \sum^T_{t=1} \eps_t^{2(\kappa-1)}.\label{eq:info_upper}
\end{align}
Combining (\ref{eq:info_lower}) and (\ref{eq:info_upper}), we have
\begin{align*}
\frac{3^{(\kappa-2)/\kappa}}{2C^2} \log \left(\frac{1}{\eps_T}\right) - \frac{5 \cdot 3^{(\kappa-2)/\kappa}}{4 C^2} \log 2 \le \sum^T_{t=1} \eps^{2(\kappa-1)}_t.
\end{align*}
An inequality like this must hold for all $T \ge T_0$. Lemma~\ref{lm:fun_rec} then states that there exists an infinite subsequence of times $1 \le t_1 < t_2 < \ldots$, such that $\eps_{t_j} = \Omega\left(t^{-1/(2\kappa-2)}_j\right)$, or, equivalently, that $r_{t_j} = \Omega\left(t^{-\kappa/(2\kappa-2)}_j\right)$. Since by hypothesis $r_t = O(t^{-\gamma})$, we must have $\gamma \le \kappa/(2\kappa-2)$.

When $\kappa=1$, from (\ref{eq:info_bound_2}) we have $I(M,X_t; Y_t) \le 4C^2$ for all $t$. This, together with (\ref{eq:info_lower}), gives
\begin{align*}
\frac{1}{6C^2} \log \left(\frac{1}{\eps_T}\right) - \frac{5}{12}\log 2 \le T, \qquad \forall T \ge T_0.
\end{align*}
which gives $\eps_T = \Omega(2^{-6C^2T})$ and $r_T = \Omega(2^{-6C^2T})$.
\end{proof}

\section{Concluding remarks}

Sequential optimization algorithms operating in the presence of uncertainty must be able to accumulate information in order to reduce uncertainty. As we have shown in this paper, there are fundamental limitations on the rate at which this uncertainty can be reduced, depending on the richness of the class of objective functions faced by the algorithm, the noisiness and the structure of the oracle that supplies information to the algorithm, and the manner in which the algorithm may approach the optimum (i.e.,~monotonically or not). In order to derive these fundamental limitations, we have developed a comprehensive information-theoretic machinery that makes use of the fact (which we have proved) that the problem of sequential optimization is, in a certain sense, at least as hard as hypothesis testing with feedback (or with controlled observations). This observation then leads to quantitative estimates that relate the minimum number of oracle queries needed to achieve a given level of accuracy to the overall reduction of uncertainty about the objective function being optimized. The latter is measured by the {\em mutual information} between the random choice of the objective and the history of algorithm's queries and oracle's responses. Carefully taking into account all the Markovian structures that are imposed by the sequential and the adaptive nature of the algorithm, we can obtain different upper bounds on this mutual information.

Using this machinery, we have derived tight lower bounds in several settings in optimization, both for arbitrary and for anytime optimization algorithms (in some cases improving upon existing results), and beyond, e.g.,~for experimental design and active learning. One promising direction for future work is to consider algorithms with {\em query costs}, i.e.,~when issuing each query incurs a cost that may depend on the query, and the goal is to balance the total cost of querying with the final optimization error. Recent work by Naghshvar and Javidi \cite{Tara} considers a hypothesis testing problem of this kind by relating it to optimal stopping for a Markov decision process, and the techniques developed in that work may be useful for deriving information-theoretic lower bounds for optimization problems with query costs.

%\newpage

\begin{appendices}
	\renewcommand{\theequation}{\Alph{section}.\arabic{equation}}
	\setcounter{lemma}{0}
	\setcounter{equation}{0}
	\renewcommand{\thedefinition}{\Alph{section}.\arabic{definition}}
	\renewcommand{\theproposition}{\Alph{section}.\arabic{proposition}}
	\setcounter{proposition}{0}
	\setcounter{definition}{0}

	\renewcommand{\thelemma}{\Alph{section}.\arabic{lemma}}
	
	\section{Finite-step vs.\ strong infinite-step algorithms}
	\label{app:finite_vs_infinite}
	
As we pointed out in Section~\ref{sec:IBC}, our definition of an infinite-step algorithm is somewhat restrictive, as it allows only the algorithms that use their most recently computed candidate minimizer as the next query. The following definition removes this restriction:
	
	\begin{definition}A {\em strong infinite-step algorithm} for a problem class $\Pro = (\cX,\Fun,\Ora)$ is a sequence of mappings $\tilde{\Alg} = \{\tilde{\Alg}_t : \cX^{t-1} \times \cY^{t-1} \to \cX \times \cX \}^{\infty}_{t=1}$. The set of all infinite-step algorithms for $\Pro$ will be denoted by $\tilde{\Algs}_\infty(\Pro)$.
	\end{definition}
	\noindent The interaction of any $\tilde{\Alg} \in \tilde{\Algs}_\infty(\Pro)$  with $\Ora$ is described recursively as follows:

	\begin{enumerate}

	\item At time $t=0$, a problem instance $f \in \Fun$ is selected by Nature and revealed to $\Ora$, but not to $\tilde{\Alg}$.

	\item At each time $t = 1,2,\ldots$:
	\begin{itemize}
	\item  $\tilde{\Alg}$ computes
	$$
	(X_t,\wh{X}_t) = \tilde{\Alg}_t(X^{t-1},Y^{t-1}),
	$$
	where $X_\tau$ and $\wh{X}_\tau$ are, respectively, the query and the candidate minimizer at time $\tau$.
	\item $\Ora$ responds with a random element $Y_t \in \cY$ according to $P( dY_t | f,X_t)$.
	\end{itemize}
	\end{enumerate}
In other words, both $\wh{X}_t$, the candidate minimizer at time $t$, and $X_t$, the query at time $t$, are computed on the basis of all currently available data, i.e., $(X^{t-1},Y^{t-1})$, yet the algorithm has more freedom, since at time $t+1$ it can query the oracle with an arbitrary point, rather than just $\wh{X}_t$. The error of $\tilde{\Alg}$ on $f \in \Fun$ at time $t$ is given by
$$
\err_{\tilde{\Alg}}(t,f) = f(\wh{X}_t) - \inf_{x \in \cX}f(x) = f(\wh{X}_{t}) - f^*.
$$

	\begin{definition} Fix a problem class $\Pro = (\cX,\Fun,\Ora)$. For any $r \ge 1$, $\eps > 0$, and $\delta \in (0,1)$, we define the $r$th-order infinite-step {\em $(\eps,\delta)$-complexity} and the {\em $\eps$-complexity} of $\Pro$, respectively, as
	\begin{align*}
	& K^{(r),\infty}_\Pro(\eps,\delta) \deq \inf \Big\{ T \ge 1: \exists \tilde{\Alg} \in \tilde{\Algs}_\infty(\Pro) \nonumber\\
	& \qquad \qquad \emph{\,\,s.t.\,\,}\sup_{f \in \Fun} \Pr \big( \err^r_{\tilde{\Alg}}(t,f) \ge \eps \big) \le \delta, \forall t > T \Big\} ; \\
	& K^{(r),\infty}_\Pro(\eps) \deq \inf \Big\{ T \ge 1: \exists {\tilde{\Alg}} \in \tilde{\Algs}_\infty(\Pro) \nonumber\\
	& \qquad \qquad \emph{\,\,s.t.\,\,}\sup_{f \in \Fun} \E \err^r_{\tilde{\Alg}}(t,f) < \eps, \forall t > T \Big\}.
	\end{align*}
	\end{definition}

	\noindent It turns out that these notions of complexity are equivalent to the ones introduced earlier:

	\begin{proposition}\label{prop:inf_and_finite_complexity} For any problem class $\Pro$ and all $r \ge 1$, $\eps > 0$, $\delta \in (0,1)$, we have
	\begin{align}
	K^{(r),\infty}_\Pro(\eps,\delta) &= K^{(r)}_\Pro(\eps,\delta) \label{eq:finite_vs_anytime} \\
	K^{(r),\infty}_\Pro(\eps) &= K^{(r)}_\Pro(\eps). \label{eq:finite_vs_anytime_2}
	\end{align}
	\end{proposition}

\begin{IEEEproof}	We only prove \eqref{eq:finite_vs_anytime}, since the proof of \eqref{eq:finite_vs_anytime_2} is similar. Likewise, we will only consider the $r=1$ case.

	First we prove that $K^\infty_\Pro(\eps,\delta) \le K_\Pro(\eps,\delta)$. We can assume that $K_\Pro(\eps,\delta) < \infty$, for otherwise the inequality holds {\em a fortiori}. Given $\eps$ and $\delta$, consider any $T$ for which there exists some $T$-step algorithm $\Alg \in \Algs_T(\Pro)$, such that
	\begin{align*}
	\sup_{f \in \Fun} \Pr\big( \err_\Alg(T,f) \ge \eps \big) \le \delta.
	\end{align*}
	Given $\Alg$, we can construct a strong infinite-step algorithm $\tilde{\Alg} \in \tilde{\Algs}_\infty(\Pro)$ as follows. Choose an arbitrary $T$-tuple $(\wh{x}_1,\ldots,\wh{x}_T) \in \cX^T$ and let
	\begin{align*}
	& \tilde{\Alg}_t(x^{t-1},y^{t-1}) \\
	&= \begin{cases} (\Alg_t(x^{t-1},y^{t-1}),\wh{x}_t),  & t = 1,\ldots,T; \\
	 (\Alg_{T+1}(x^{T},y^{T}),\Alg_{T+1}(x^T,y^T)), & t > T.
	\end{cases}
	\end{align*} 
Then it's clear that for any $t > T$
	\begin{align*}
&	\sup_{f \in \Fun} \Pr\big( \err_{\tilde{\Alg}}(t,f) \ge \eps \big)  \nonumber\\
&\qquad = \sup_{f \in \Fun} \Pr\big( f(\wh{X}_{t}) - f^* \ge \eps\big) \\
&\qquad = \sup_{f \in \Fun} \Pr\big( f(\Alg_{T+1}(X^T,Y^T)) - f^* \ge \eps \big) \\
&\qquad = \sup_{f \in \Fun} \Pr\big( \err_{\Alg}(T,f) \ge \eps \big) \\
&\qquad \le \delta.
	\end{align*}
	Hence, $K^\infty_\Pro(\eps,\delta) \le K_\Pro(\eps,\delta)$.

	Next, we prove $K_\Pro(\eps,\delta) \le K^\infty_\Pro(\eps,\delta)$. Again, we can assume that $K^\infty_\Pro(\eps,\delta) < \infty$. Consider an algorithm $\tilde{\Alg} \in \tilde{\Algs}_\infty(\Pro)$, such that
	\begin{align}\label{eq:infinite_assumption}
	\sup_{t > T} \sup_{f \in \Fun} \Pr \big( \err_{\tilde{\Alg}}(t,f) \ge \eps \big) \le \delta
	\end{align}
	for some $T$. Let $\Pi_1$ and $\Pi_2$ denote the two coordinate projection mappings from $\cX \times \cX$ onto $\cX$, i.e., $\Pi_1(x,x') = x$ and $\Pi_2(x,x') = x'$, and define $\Alg \in \Algs_{T}(\Pro)$ by setting
\begin{align*}
	\Alg_t = \begin{cases}
	\Pi_1 \circ \tilde{\Alg}_t, & t = 1,\ldots,T \\
	\Pi_2 \circ \tilde{\Alg}_{T+1}, & t = T+1.
\end{cases}
\end{align*}
Then from \eqref{eq:infinite_assumption}
	\begin{align*}
	& \sup_{f \in \Fun} \Pr \big( \err_\Alg(T,f) \ge \eps \big) \\
	& \qquad = \sup_{f \in \Fun} \Pr \big( f(X_{T+1}) - f^* \ge \eps\big) \\
	& \qquad = \sup_{f \in \Fun} \Pr \big( f(\Alg_{T+1}(X^T,Y^T)) - f^* \ge \eps\big) \\
	& \qquad = \sup_{f \in \Fun} \Pr \big(\err_{\tilde{\Alg}}(T+1,f) \ge \eps\big) \\
	& \qquad \le \delta,
	\end{align*}
	which implies $K_\Pro(\eps,\delta) \le K^\infty_\Pro(\eps,\delta)$.
\end{IEEEproof}

\section{Miscellaneous proofs}
\label{app:proofs}

\subsection{Proof of Proposition~\ref{prop:prob_det_bounds}}

Given $\eps$ and $\Pro$, consider any $T$ for which there exists some algorithm $\Alg \in \Algs_T(\Pro)$ that satisfies $\sup_{f \in \Fun} \E \err^r_\Alg(T,f) \le \eps$. Then Markov's inequality gives
\begin{align*}
\Pr \big( \err^r_\Alg(T,f) \ge \eps/\delta \big) \le \frac{\E \err^r_\Alg(T,f)}{\eps/\delta} \le \delta, \quad \forall f \in \Fun.
\end{align*}
Hence, $T \ge K^{(r)}_\Pro(\eps/\delta,\delta)$. Taking the infimum over all such $T$, we arrive at the proof.

\subsection{Proof of Lemma~\ref{lm:oblivious_info_bounds}}

First, we modify the construction of the probability space $(\Omega,\cB,\Prob)$ in Section~\ref{sec:HTF} by introducing the random variables $U^T \in \cU^T$ that describe the responses of the ``clean'' (deterministic) oracle $\psi : \Fun \times \cX \to \cU$ to the queries $X^T$. The relevant causal ordering is
\begin{align*}
	M,X_1,U_1,Y_1,\ldots,X_t,U_t,Y_t,\ldots,X_T,U_T,Y_T,X_{T+1},
\end{align*}
where, $\Prob$-almost surely, we have \eqref{eq:Nature_choice} and
\begin{align*}
	&\Prob(X_t \in A | M, X^{t-1},U^{t-1},Y^{t-1}) = \ind{\Alg_t(X^{t-1},Y^{t-1}) \in A}\\
	&\Prob(U_t \in C | M, X^t,U^{t-1},Y^{t-1})   = \ind{\psi(f_M,X_t) \in C} \\
	&\Prob(Y_t \in B | M, X^t, U^t, Y^{t-1}) = Q(B|U_t)
\end{align*}
for all $A \in \cB_\cX, B \in \cB_\cY, C \in \cB_\cU$. That is, $(M,U^{t-1}) \to (X^{t-1},Y^{t-1}) \to X_t$, $(X^{t-1},U^{t-1},Y^{t-1}) \to (M,X_t) \to U_t$, and $(M,X^t,U^{t-1},Y^{t-1}) \to U_t \to Y_t$ are Markov chains for each $t$. Then we can write
\begin{align*}
I(M; Y_t|X^t,Y^{t-1}) = I(M, U_t ; Y_t | X^t, Y^{t-1})
\end{align*}
because $U_t$ is completely determined by $M$ and $X_t$ via $U_t = \psi(f_M,X_t)$. Moreover,
\begin{align*}
& I(M, U_t ; Y_t | X^t, Y^{t-1}) \nonumber\\
& = I(U_t ; Y_t | X^t, Y^{t-1}) + I(M; Y_t | U_t, X^t, Y^{t-1}) \\
& = I(U_t; Y_t | X^t, Y^{t-1}),
\end{align*}
where the first step is by the chain rule and the second step is due to the fact that $M \to (U_t,X^t,Y^{t-1}) \to Y_t$ is a Markov chain. This follows by applying the weak union and the decomposition properties of conditional independence \cite[p.~11]{Pearl} to the Markov chain $(M,X^t,U^{t-1},Y^{t-1}) \to U_t \to Y_t$. By the same token, $(X^t,Y^{t-1}) \to U_t \to Y_t$ is also a Markov chain, so we have $I(U_t ; Y_t | X^t , Y^{t-1}) \le I(U_t; Y_t)$.

\subsection{Proof of Lemma~\ref{lm:refined_info_bounds}}

Let us fix some $t$ and consider the conditional mutual information term $I(M; Y_t |X^t,Y^{t-1})$ in the summation in Lemma~\ref{lm:info_bounds}:
\begin{align}
&I(M; Y_t | X^t,Y^{t-1}) \nonumber\\
&= D(\Prob_{Y_t|M,X^t,Y^{t-1}} \| \Prob_{Y_t|X^t,Y^{t-1}} | \Prob_{M,X^t,Y^{t-1}}) \label{eq:div_bound_0}\\
&= \E \left[ \log \frac{d\Prob_{Y_t|M,X^t,Y^{t-1}}}{d\Prob_{Y_t|X^t,Y^{t-1}}} \right] \nonumber \\
&= \E \left[ \log \frac{d\Prob_{Y_t|M,X^t,Y^{t-1}}}{d\ProbQ_{Y_t|X^t,Y^{t-1}}}\right] - \E \left[ \log \frac{d\Prob_{Y_t|X^t,Y^{t-1}}}{d\ProbQ_{Y_t|X^t,Y^{t-1}}}\right] \label{eq:div_bound_1} \\
&= D\big(\Prob_{Y_t|M,X^t,Y^{t-1}} \big\| \ProbQ_{Y_t | X^t,Y^{t-1}} \big| \Prob_{M,X^t,Y^{t-1}} \big) \nonumber\\
& \qquad \qquad - D\big( \Prob_{Y_t|X^t,Y^{t-1}} \big\| \ProbQ_{Y_t|X^t,Y^{t-1}} \big| \Prob_{X^t,Y^{t-1}} \big) \label{eq:div_bound_2} \\
& \le D\big(\Prob_{Y_t|M,X^t,Y^{t-1}} \big\| \ProbQ_{Y_t | X^t,Y^{t-1}} \big| \Prob_{M,X^t,Y^{t-1}} \big), \label{eq:div_bound_3}
\end{align}
where \eqref{eq:div_bound_0} follows from \eqref{eq:cond_MI}, \eqref{eq:div_bound_1} and \eqref{eq:div_bound_2} are justified by virtue of \eqref{eq:AC_condition}, while \eqref{eq:div_bound_3} follows from the fact that the divergence is nonnegative.

\subsection{Proof of Lemma~\ref{lm:1st_order_IR_bound}}

The random variables $V^0 = f_M(X) + W$ and $V^1 = g_M(X) + Z$ are conditionally independent given $M=i$ and $X=x$:
\begin{align*}
\Prob_{Y|M=i,X=x} = \Prob_{V^0|M=i,X=x} \otimes \Prob_{V^1|M=i,X=x},
\end{align*}
where
\begin{align*}
\Prob_{V^0|M=i,X=x} &= \Normal(f_i(x),\sigma^2) \\
\Prob_{V^1|M=i,X=x} &= \Normal(g_i(x),\sigma^2 I_n).
\end{align*}
Therefore,
\begin{align*}
& D(\Prob_{Y|M=i,X=x} \| \Prob_{Y|M=j,X=x}) \nonumber\\
&= D(\Prob_{V^0|M=i,X=x} \| \Prob_{V^0|M=j,X=x}) \nonumber\\
& \qquad \qquad + D(\Prob_{V^1|M=i,X=x} \| \Prob_{V^1|M=j,X=x}) \\
&= D\big(\Normal(f_i(x),\sigma^2) \big\| \Normal(f_{j}(x),\sigma^2)\big) \nonumber\\
& \qquad \qquad + D\big(\Normal(g_i(x),\sigma^2 I_n) \big\| \Normal(g_{j}(x),\sigma^2 I_n)\big) \\
&= \frac{1}{2\sigma^2} \left\{ \left[ f_i(x) - f_{j}(x) \right]^2 + \| g_i(x) - g_{j}(x) \|^2 \right\}.
\end{align*}
Plugging this into \eqref{eq:IR_bound}, we get \eqref{eq:1st_order_IR_bound}.

\subsection{Proof of Lemma~\ref{lm:LF_SC}}

Let $Q^*$ denote the product normal distribution $\Normal(c^*,\sigma^2) \otimes \Normal(0,\sigma^2 I_\Dim)$. Observe that, for every $i=0,1,\ldots,N-1$,
\begin{align*}
Y^{(i)} = \big(f_i(x^*_i) + W, \nabla f_i(x^*_i) + Z\big) = (c^* + W,Z) \sim Q^*.
\end{align*}
Let $X_t$ denote the query of $\Alg$ at time $t$ and let $Y_t$ be the corresponding oracle response. Then
\begin{align*}
\Prob_{Y_t|M=i,X_t=x_t} = \Normal(f_i(x_t),\sigma^2) \otimes \Normal(\nabla f_i(x_t),\sigma^2 I_\Dim).
\end{align*}
%and consequently, letting $\ProbQ^*_{Y_t}$ denote an independent copy of $Q^*$,
%\begin{align}
%D(\Prob_{Y_t|M,X_t} \| \ProbQ^*_{Y_t} | \Prob_{M,X_t}) = \frac{1}{2\sigma^2} \E \left\{ \left[f_M(X_t) - c^* \right]^2 + \| \nabla f_M(X_t) \|^2 \right\}.
%\end{align}
Hence, applying the LF bound \eqref{eq:LF_bound} with $\ProbQ^*_{Y_t} = Q^*$, we can write
\begin{align}\label{eq:LF_bound_SC}
& I(M; Y_t | X^t, Y^{t-1}) \nonumber\\
&  \le \frac{1}{2\sigma^2}  \max_i \E \left\{ \left[ f_i(X_t) - c^* \right]^2 + \| \nabla f_i(X_t) \|^2 \right\}.
\end{align}
We now relate the right-hand side of \eqref{eq:LF_bound_SC} to the performance of $\Alg$. First of all, by convexity of $f_i$,
\begin{align}
f_i(X_t) - c^* &= f_i(X_t) - f_i(x^*_i) \nonumber\\
&\le \nabla f_i(X_t)^\tr (X_t - x^*_i) \nonumber\\
& \le \| \nabla f(X_t) \|  \| X_t - x^*_i \| \nonumber\\
&\le L D_\cX \| X_t - x^*_i \|, \label{eq:SC_DR_1}
\end{align}
where in the last step we have used the fact that
\begin{align*}
\| \nabla f_i(X_t) \| = \| \nabla f_i(X_t) - \nabla f_i(x^*_i) \|  \le L \| X_t - x^*_i \| \le L D_\cX.
\end{align*}
On the other hand, from  strong convexity we have that
\begin{align}\label{eq:SC_DR_2}
f_i(X_t) \ge c^* + (\kappa^2/2) \| X_t - x^*_i \|^2.
\end{align}
Combining \eqref{eq:SC_DR_1} and \eqref{eq:SC_DR_2}, we therefore obtain
\begin{align}
\left[ f_i(X_t) - c^* \right]^2 & \le 2 D^2_\cX (L/\kappa)^2[f_i(X_t) - c^* ] \nonumber \\
&= 2 D^2_\cX (L/\kappa)^2 \err_\Alg(t,f_i). \label{eq:SC_DR_3}
\end{align}
Moreover, because $\nabla f_i(x^*_i) = 0$,  we can write
\begin{align}
\| \nabla f_i(X_t) \|^2 &=  \| \nabla f_i(X_t) - \nabla f_i(x^*_i) \|^2  \nonumber \\
&\le L^2  \| X_t - x^*_i \|^2 \nonumber \\
&\le 2(L/\kappa)^2 
\err_\Alg(t,f_i).\label{eq:SC_DR_4}
\end{align}
Substituting  \eqref{eq:SC_DR_3} and \eqref{eq:SC_DR_4} into \eqref{eq:LF_bound_SC}, we get (\ref{eq:LF_SC}). Eq.~\eqref{eq:dim_returns} is immediate from definitions.

\section{Lemma on functional recurrences}

\begin{lemma}\label{lm:fun_rec} Suppose that $\{\eps_t\}$ is a sequence of nonnegative reals satisfying 
\begin{align*}
K \log \left(\frac{1}{\eps_T}\right) - L \le \sum^T_{t=1} \eps^\alpha_t, \qquad \forall T
\end{align*}
for some $K,L,\alpha > 0$. Then there exists some constant $0 < c < (K/\alpha)^{1/\alpha}$, such that $\eps_t \ge ct^{-1/\alpha}$ for infinitely many values of $t$. \end{lemma}

\begin{proof}The proof is by contradiction. Suppose first that $t^{1/\alpha}\eps_t < c$ for all $t$. Then
\begin{align*}
K \log \frac{T^{1/\alpha}}{c} - L < c^\alpha \sum^T_{t=1} \frac{1}{t} \le c^\alpha (\log T + 1), \,\, \forall T.
\end{align*}
Rearranging, we get
\begin{align*}
\left(\frac{K}{\alpha} - c^\alpha\right) \log T \le K \log c + L + c^\alpha, \qquad \forall T.
\end{align*}
Since $K/\alpha - c^\alpha > 0$, this implies that $\log T$ is bounded for all large and positive $T$, which is impossible. Hence there exists some set $S \subseteq \Naturals$, such that $\eps_t \ge ct^{-1/\alpha}$ for all $t \in S$. We now show that $S$ must necessarily be countably infinite. Suppose, to the contrary, that it's finite.  Then there is some $T_0$, such that $\eps_t < ct^{-1/\alpha}$ for all $t \ge T_0$. In that case, for $T > T_0$, we can write
\begin{align*}
K \log \frac{T^{1/\alpha}}{c}  - L &\le c^\alpha \left(\sum^{T_0}_{t=1} \eps^\alpha_t + \sum^T_{t=T_0 + 1} \frac{1}{t}\right) \\
&\equiv c^\alpha \left( K(T_0,\alpha) + \sum^T_{t=T_0+1} \frac{1}{t}\right) \\
& \le c^\alpha \left(K(T_0,\alpha) + \log T - \log T_0 \right).
\end{align*}
Rearranging, we see that the inequality
\begin{align*}
& \left(\frac{K}{\alpha} - c^\alpha\right) \log T \\
& \qquad \qquad \le K\log c + L + c^\alpha(K(T_0,\alpha) - \log T_0)
\end{align*}
must hold for all $T > T_0$. Since $K/\alpha > c^\alpha$ by hypothesis, this implies that $\log T$ is bounded for $T > T_0$, which is, again, impossible. Thus, $\eps_t \ge ct^{-1/\alpha}$ for infinitely many values of $t$. \end{proof}

\end{appendices}

\section*{Acknowledgment}

The authors would like to thank the anonymous reviewers for their probing questions and numerous suggestions, which have greatly improved the paper. In particular, we would like to thank one reviewer for suggesting the definition of a strong infinite-step algorithm.

\bibliography{opt_feedback_IT.bbl}

\begin{IEEEbiographynophoto}{Maxim Raginsky} (S'99--M'00) received the B.S. and M.S. degrees in 2000 and the
Ph.D. degree in 2002 from Northwestern University, Evanston, IL, all in electrical engineering. From 2002 to 2004 he was a Postdoctoral
Researcher at the Center for Photonic Communication and Computing at
Northwestern University, where he pursued work on quantum cryptography
and quantum communication and information theory. From 2004 to 2007 he
was a Beckman Foundation Postdoctoral Fellow at the University of
Illinois in Urbana-Champaign, where he carried out research on
information theory, statistical learning and computational
neuroscience. Since September 2007 he has been with Duke University, where he is now Assistant Research Professor of Electrical and Computer Engineering. His interests include statistical signal processing,
information theory, statistical learning and nonparametric
estimation. He is particularly interested in problems that combine the
communication, signal processing and machine learning components in a
novel and nontrivial way, as well as in the theory and practice
of robust statistical inference with limited information.
\end{IEEEbiographynophoto}

\begin{IEEEbiographynophoto}{Alexander Rakhlin} received the B.A. degrees in mathematics and
computer science from Cornell University in 2000, and Ph.D. in
computational neuroscience from MIT in 2006. From 2006 to 2009 he was
a postdoctoral fellow at the Department of Electrical Engineering and
Computer Sciences, UC Berkeley. Since 2009, he has been on the faculty
in the Department of Statistics, University of Pennsylvania. He has
been a co-director of Penn Research in Machine Learning (PRiML) since
2010. His interests include machine learning, statistics,
optimization, and game theory.
\end{IEEEbiographynophoto}

\end{document}